\documentclass[pra,aps,superscriptaddress,twocolumn,letter,nopacs,nofootinbib,longbibliography]{revtex4-2}

\usepackage{array,multirow,graphicx,color,amsmath,amsfonts,enumerate,amsthm,amssymb,mathtools,enumitem,thmtools,hyperref,subfigure,mathdots,enumitem,centernot,bm,soul,bbm,tikz,pgfplots,enumitem}
\usepackage[capitalise, noabbrev]{cleveref}
\usetikzlibrary{arrows}
\pgfplotsset{compat=1.14}
\hypersetup{colorlinks=true,linkcolor=blue,citecolor=blue,urlcolor=blue}

\usepackage{tcolorbox}\tcbset{before skip=10pt,toptitle=2mm,bottomtitle=1mm,fonttitle=\bfseries}
\tcbuselibrary{theorems}
\tcbuselibrary{breakable}
\definecolor{NavyBlue}{rgb}{0.0, 0.0, 0.5}
\definecolor{OliveGreen}{rgb}{0.33, 0.42, 0.18}
\definecolor{niceRed}{rgb}{0.825, 0.248, 0.248}
\definecolor{niceBlue}{rgb}{0.248, 0.248, 0.825}
\definecolor{niceGreen}{rgb}{0.125, 0.6, 0.125}
\definecolor{def_color_frame}{RGB}{220,230,242}
\colorlet{def_color_back}{def_color_frame!30}
\definecolor{def_color_text}{RGB}{37,64,97}
\newtcolorbox[auto counter]{pabox}[2][]{%
	colback=def_color_back,colframe=def_color_frame,fonttitle=\bfseries,coltitle=def_color_text,breakable,title=#2,#1}

\def\H{ {\cal H} }

\def\L{ {\cal L} }

\def\>{\rangle}
\def\<{\langle}
\newcommand{\bra}[1]{\langle {#1} |}
\newcommand{\ket}[1]{| {#1} \rangle}
 
\newcommand{\ketbra}[2]{\ensuremath{\left|#1\right\rangle\!\!\left\langle#2\right|}}

\newcommand{\tr}[1]{\mathrm{Tr}\left( #1 \right)}

\let\vv\v
\renewcommand{\v}[1]{\ensuremath{\boldsymbol #1}}

\definecolor{ppblue}{RGB}{46,117,182}
\definecolor{ppred}{RGB}{197, 90, 17}

\theoremstyle{plain}
\newtheorem{thm}{Theorem}

\newtheorem{lem}[thm]{Lemma}

\newtheorem{cor}[thm]{Corollary}

\theoremstyle{definition}
\newtheorem{defn}{Definition}

\definecolor{tikzBlue}{rgb}{0.6941176470588235,0.7568627450980392,0.8588235294117647}
\definecolor{tikzOrange}{rgb}{0.9294117647058824,0.7647058823529411,0.49019607843137253}
\definecolor{tikzBlue2}{rgb}{0.462745098,0.504575163,0.57254902}
\definecolor{tikzOrange2}{rgb}{0.619607843,0.509803922,0.326797386}
\definecolor{tikzGray}{rgb}{0.7529411764705882,0.7529411764705882,0.7529411764705882}

\DeclareFontFamily{U}{mathb}{\hyphenchar\font45}
\DeclareFontShape{U}{mathb}{m}{n}{
	<-6> mathb5 <6-7> mathb6 <7-8> mathb7
	<8-9> mathb8 <9-10> mathb9
	<10-12> mathb10 <12-> mathb12
}{}
\DeclareSymbolFont{mathb}{U}{mathb}{m}{n}
\DeclareMathSymbol{\llcurly}{\mathrel}{mathb}{"CE}
\DeclareMathSymbol{\ggcurly}{\mathrel}{mathb}{"CF}

\renewcommand*{\thefootnote}{\fnsymbol{footnote}}

\begin{document}

\title{Continuous thermomajorization and a complete set of laws \\for Markovian thermal processes}

\author{Matteo Lostaglio$^{*,}$}
\affiliation{Korteweg-de Vries Institute for Mathematics and QuSoft, University of Amsterdam, The Netherlands}
\affiliation{QuTech, Delft University of Technology, P.O. Box 5046, 2600 GA Delft, The Netherlands}

\author{Kamil Korzekwa$^{*,}$}
\affiliation{Faculty of Physics, Astronomy and Applied Computer Science, Jagiellonian University, 30-348 Krak\'{o}w, Poland}

\footnotetext{These authors contributed equally to this work. \\ Emails: lostaglio@protonmail.com, korzekwa.kamil@gmail.com}

\begin{abstract}
	The standard dynamical approach to quantum thermodynamics is based on Markovian master equations describing the thermalization of a system weakly coupled to a large environment, and on tools such as entropy production relations. Here we develop a new framework overcoming the limitations that the current dynamical and information theory approaches encounter when applied to this setting. More precisely, we introduce the notion of continuous thermomajorization, and employ it to obtain necessary and sufficient conditions for the existence of a Markovian thermal process transforming between given initial and final energy distributions of the system. These lead to a complete set of generalized entropy production inequalities including the standard one as a special case. Importantly, these conditions can be reduced to a finitely verifiable set of constraints governing non-equilibrium transformations under master equations. What is more, the framework is also constructive, i.e., it returns explicit protocols realizing any allowed transformation. These protocols use as building blocks elementary thermalizations, which we prove to be universal controls. Finally, we also present an algorithm constructing the full set of energy distributions achievable from a given initial state via Markovian thermal processes and provide a \texttt{Mathematica} implementation solving $d=6$ on a laptop computer in minutes.		 
\end{abstract}

\maketitle
\renewcommand{\thefootnote}{\arabic{footnote}}
\setcounter{footnote}{0}

%------------------------------------------------------------------
% SEC. I
%------------------------------------------------------------------

\section{Introduction}
\label{sec:intro}

It is well-known that thermodynamics can be formulated in the resource theory language of information theory~\cite{lostaglio2019introductory, janzing2000thermodynamic, horodecki2013fundamental, brandao2013second, ng2018resource, vinjanampathy2016quantum}. Since it focuses only on input-output relations under a class of quantum operations, the resource theory is unable to discuss how the process is realized in time, nor does it involve the notions of entropy production or master equations, even though the latter are commonplace in standard approaches. As such, the standard and the resource-theoretic frameworks are for the most part disconnected. Somewhat surprisingly, despite all its powerful theorems, the resource theory has not affected the analysis of a thermodynamics practitioner following the more explicit dynamical approaches based on the master equation formalism~\cite{kosloff2013quantum}. 

The aim of this work is to overcome this situation by showing how to unify the master equations and information theory tools. We first summarize the basic notions of both approaches to thermodynamics in Sec.~\ref{sec:setting}, and demonstrate the insufficiency of each of them, when taken separately, to capture relevant thermodynamic constraints. Then, in Sec.~\ref{sec:approach}, we propose a hybrid approach that can overcome these limitations. In constructing our solution we highlight the necessity of finding a \emph{finitely verifiable} set of thermodynamic laws specifying when one state can be thermodynamically transformed into another one. Furthermore, we wish to go beyond the question of \emph{whether} a thermodynamic transformation exists, and ask \emph{how} it should be realized. The resource theory approaches are for the most part silent about the latter, which is arguably a central obstacle to their application to concrete problems. Our framework, on the contrary, will be \emph{constructive}. 

In Sec.~\ref{sec:thermomajo} we introduce our main novel technical tool, \emph{continuous thermomajorization}, which extends the concept of thermomajorization introduced in Refs.~\cite{ruch1978mixing,horodecki2013fundamental}. We prove that this partial order of energy distributions provides necessary and sufficient conditions for the existence of a thermalization process generated by a Markovian master equation. More precisely, we show that an out of equilibrium energy distribution can be transformed into another one by a Markovian thermal process if and only if the former continuously thermomajorizes the latter. We then connect with the notion of entropy production~\cite{landi2021irreversible} in Sec.~\ref{sec:complete_set}, where we show that continuous thermomajorization allows one to identify a \emph{complete} set of generalized entropy production relations, including the standard one as a special case.
	
Our ultimate goal, however, is to find a finitely-verifiable set of such conditions. To achieve this, we first prove in Sec.~\ref{sec:controls} that \emph{elementary thermalizations}, generated by simple reset master equations on two-level submanifolds, form universal controls in the Markovian regime. In other words, every state transformation that can be achieved by a Markovian thermal process can also be achieved by a sequence of elementary thermalizations. Employing this simplified set of controls, in Sec.~\ref{sec:second_laws} we show that continuous thermomajorization indeed can be checked in a finite number of steps. Remarkably, we can also return the exact sequence of elementary controls required to realize any allowed transformation. As a final result of our paper, in Sec.~\ref{sec:algorithm} we provide an explicit algorithm verifying the continuous thermomajorization relation between any two vectors, and offer a Mathematica implementation~\cite{korzekwa2021continuous} that checks the conditions for systems of dimension up to $7$ in a matter of hours on a standard laptop computer. 

While this work focuses on developing the technical machinery, in the accompanying paper~\cite{korzekwa2022optimizing} we illustrate how one can employ it to design provably optimal thermodynamic protocols. There, we apply our results to study the effects of memory on work fluctuations, to explicitly construct optimal cooling protocols, and to showcase the important role played by catalysts in practical scenarios. We also report on a recent work which used our results to show that non-Markovianity boosts the efficiency of thermal bio-molecular switches~\cite{spaventa2021non}. All these build up encouraging evidence that the framework is suitable both for providing model-independent bounds, as well as for algorithmically constructing new thermodynamics protocols when a complete analysis is unattainable by either analytic or numerical methods. 

%------------------------------------------------------------------
% SEC. II
%------------------------------------------------------------------

\section{Setting the scene}
\label{sec:setting}

%------------------------------------------------------------------
% SEC. II.A
%------------------------------------------------------------------

\subsection{Thermodynamic frameworks}
\label{sec:frameworks}

%------------------------------------------------------------------
% SEC. II.A.1
%------------------------------------------------------------------

\subsubsection{The traditional approach}
\label{sec:setting_standard}

A general open dynamics of a $d$-dimensional quantum system is described by a quantum channel, i.e., a completely positive trace-preserving map acting on the quantum state $\rho$~\cite{nielsen2010quantum}. However, in a thermodynamic setting, we are often interested in the evolution of such a quantum system interacting with a large thermal bath at inverse temperature~$\beta = 1/(k_B T)$, where $k_B$ is Boltzmann constant and $T$ is the temperature of the bath. Typical microscopic derivations employing the weak coupling limit~\cite{breuer2002open,kosloff2013quantum, alicki2018introduction} then lead to a master equation with the general form~\cite{gorini1976completely,lindblad1976generators}:
\begin{equation}
	\label{eq:master}
	\frac{d \rho(t)}{dt}  = \H(\rho(t)) + \mathcal{L}_t(\rho(t)).
\end{equation}
The first term, $\H$, is the generator of a closed (reversible) quantum dynamics, 
\begin{equation}
	\H(\rho)=- i [H, \rho],
\end{equation}
with $[\cdot,\cdot]$ denoting the commutator, $[A,B] = AB - BA$, and $H$ being the (dressed) Hamiltonian of the system, which we assume to be time-independent. The second term, $\mathcal{L}_t$, is known as the \emph{Lindbladian} or \emph{dissipator} and generates an open (irreversible) quantum dynamics. It has the following general form
\begin{equation}
	\label{eq:lindbladian}
	\!\!\!\mathcal{L}_t(\rho) = \sum_{i} r_i(t) \left(\! L_i(t) \rho L_i(t)^\dag - \frac{1}{2}\{L_i(t)^\dag L_i(t), \rho\} \!\right)\!,\!
\end{equation} 
with $\{\cdot,\cdot\}$ denoting the anticommutator, \mbox{$\{A,B\} = AB + BA$}, $L_i(t)$ being time-dependent \emph{jump operators}, and $r_i(t)$ being time-dependent non-negative \emph{jump rates}.

While a general Lindbladian only requires the rates $r_i$ to be non-negative, Lindbladians arising from the interaction of a quantum system with a large heat bath have two standard properties~\cite{breuer2002open,kosloff2013quantum, alicki2018introduction}:
\begin{enumerate}[label=(P\arabic*)]
	\item\label{p1}\textbf{Stationary thermal state.} The Gibbs thermal state of the system, 
	\begin{equation}
		\label{eq:thermalstate}
		\gamma = \frac{e^{-\beta H}}{\tr{e^{-\beta H}}},
	\end{equation}
	is a stationary solution of the dynamics, i.e.,
	\begin{equation}
		\forall t: \quad \mathcal{L}_t\gamma = 0. 
	\end{equation} 
	\item\label{p2}\textbf{Covariance.} The Lindbladian $\L_t$ commutes with the generator of the Hamiltonian dynamics $\H$ at all times $t$, i.e.,
	\begin{equation}
		\forall \rho:\quad\mathcal{L}_t(\mathcal{H}(\rho)) = \mathcal{H}(\mathcal{L}_t(\rho)).
	\end{equation}
\end{enumerate}

For brevity, we will refer to the quantum dynamics generated by master equations in the form of Eq.~\eqref{eq:master} and satisfying properties \ref{p1}-\ref{p2} as \emph{Markovian thermal processes}:
\begin{defn}
	A channel $\mathcal{T}$ is a \emph{Markovian thermal process} (MTP) if it results from integrating a Markovian master equation, Eq.~\eqref{eq:master}, between time $0$ and $t_f \in [0, +\infty]$, where the Lindbladian $\mathcal{L}_t$ satisfies properties~\ref{p1}-\ref{p2}.
\end{defn}
These form a standard description of thermalization in the field of quantum thermodynamics and beyond (see, e.g., Sec. 3.1 of Ref.~\cite{alicki2018introduction}), and will be the main focus of this work. The most well-known and important constraint on the allowed thermodynamic transitions under MTP takes the form of a sort of $H$-theorem~\cite{alicki2018introduction}:
\begin{align}
	\frac{d \Sigma(t)}{dt} :=& -\frac{d}{dt} S(\rho(t)\| \gamma)\geq  0, \label{eq:entropy_prod} 
\end{align}
where 
\begin{equation}
	S(\rho\| \gamma) = \tr{\rho (\log \rho - \log \gamma)}
\end{equation}
is the quantum relative entropy and $d \Sigma/dt$ is known as the \emph{entropy production} (relative to $\gamma$)~\cite{spohn1978entropy}. One can recognize in this equation the standard second law of thermodynamics:
\begin{equation}
	\label{eq:standard_entropy_prod}
	\frac{d \Sigma(t)}{dt} =	\frac{dS(t)}{dt} - \beta J(t) \geq 0,  
\end{equation}
with
\begin{subequations}
	\begin{align}
		S(t) &:= - \tr{\rho(t) \log \rho(t)},\\
		J &:= \tr{H \frac{d\rho(t)}{dt}},
	\end{align}
\end{subequations}
being the von Neumann entropy and the heat current flowing to the system, respectively.  

A typical example of a Markovian thermal process is the quantum optical master equation~\cite{breuer2002open}, which describes the evolution of a two-dimensional quantum system with \mbox{$H = \frac{\hbar \omega}{2}(\ketbra{1}{1}-\ketbra{0}{0})$} interacting weakly with a thermal radiation field. The quantum optical master equation is of the form given in Eq.~\eqref{eq:master}, with the jump operators and the corresponding rates given by
\begin{subequations}
	\begin{align}
		L_1&=\ketbra{0}{1},\quad r_1(t)=r(N+1),&\\
		L_2&=\ketbra{1}{0},\quad r_2(t)=rN,&		
	\end{align}
\end{subequations}
where $N$ is the average number of photons at the resonant frequency $\omega$ and $r$ is the spontaneous emission rate. 

While it is not straightforward to characterize \emph{all} physical setups modelled via MTPs, they are commonplace. A crucial observation to be kept in mind is the following: MTPs constitute an \emph{effective model} that emerges after common approximations (such as the secular approximation, ignoring the Lamb shift) and in an appropriate frame (in general, a rotating frame). For a concrete example of such an emergence after relevant approximations, consider a typical setup of discrete quantum heat engines in the weak coupling regime. There, one assumes that the dissipators only operate for a given amount of time, and they are then suddenly switched to new ones according to some schedule. This is a special case of Eq.~\eqref{eq:lindbladian} where $L_i(t) \equiv L_i$, $r_i(t)$ are taken to be appropriate step functions and $H$ is (approximatively) constant after ignoring a small Lamb shift.\footnote{Just to make an example among many, Ref.~\cite{uzdin2015equivalence} involves this kind of ``implicit'' time dependence on the dissipators only (after rotating wave approximation). Often these dissipators are associated to baths at different temperatures coupled to different energy submanifolds, but nothing prevents us to consider the special case where the baths are at the same temperature.} These are not only a nice class of examples: as we show in Theorem~\ref{thm:universality}, these controls are sufficient to realize every transformation that can be achieved by a generic MTP.

Thus, any model that can be formally written as Eq.~\eqref{eq:master} with properties \ref{p1}-\ref{p2} (with or without time-dependence on $\mathcal{L}_t$) falls within the scope of this work. These include incoherent noise in quantum computers~\cite{campbell2017roads} and effective models describing fluorescence and other non-radiative decay channels in atoms, molecules and nanostructures. The formal equivalence between thermodynamic and other models of dissipation is well-known and it is in fact leveraged as a standard technique to realize effective heat baths~\cite{klatzow2019experimental}. In quantum information terms, depolarization and amplitude damping can be seen as limiting cases of Markovian thermal processes when \mbox{$\beta \rightarrow 0$} and \mbox{$\beta \rightarrow \infty$}, respectively. This means that, in principle, our results apply well-beyond the obvious thermodynamic scenarios. Here, we study Markovian thermal processes independently of what application one has in mind. In the accompanying paper~\cite{korzekwa2022optimizing}, we present examples of how the formalism can be applied in practice.

Having said that, there are of course several relevant scenarios that fall beyond the scope of MTPs. We will discuss the potential and need for further generalizations of the MTP framework after discussing the resource-theory approach.

%------------------------------------------------------------------
% SEC. II.A.2
%------------------------------------------------------------------

\subsubsection{The resource-theoretic approach}
\label{sec:setting_resource}

The notion of a Markovian thermal process should be contrasted with the notion of a thermal operation~\cite{horodecki2013fundamental} or the closely related notion of a thermal process~\cite{gour2018quantum}, used in the resource theory of quantum thermodynamics: 

\begin{defn}
	\label{defn:tp}
Thermal processes (TP) are all channels $\mathcal{E}$ that satisfy the two conditions analogous to~\ref{p1}-\ref{p2}: 
\begin{subequations}
	\begin{align}
		\mathcal{E}(\gamma)& = \gamma,\\
		\forall \rho,t: \quad\mathcal{E}(e^{-iHt}\rho e^{iHt}) &= e^{-iHt} \mathcal{E}(\rho) e^{iHt},
	\end{align}	
\end{subequations}
with $\gamma$ defined in Eq.~\eqref{eq:thermalstate}.
\end{defn}
 Examples of thermal processes include all \emph{thermal operations}, i.e. dynamics induced by generic energy-preserving unitary interactions between the system and an environment~$E$ described by an \emph{arbitrary} Hamiltonian $H_E$ and prepared in a thermal Gibbs state $\gamma_E$ at fixed inverse temperature~$\beta$~\cite{brandao2011resource, brandao2013second, lostaglio2019introductory}. When such a transformation exists between the initial and final states, $\rho(0)$ and $\rho(t_f)$, we will write
\begin{equation}
\label{eq:tp}
\rho(0) \stackrel{ \textrm{TP}}{\longmapsto} \rho(t_f).
\end{equation}

The resource theory is concerned with giving necessary and sufficient conditions for Eq.~\eqref{eq:tp}. Because of the symmetries inherent in the thermodynamic processes under consideration, it is convenient to represent the state of the system $\rho(t)$ at time $t$ by a vector $\v{p}(t)$ of \emph{populations} (energy distributions) and a matrix $C(t)$ of \emph{coherences}, defined as
\begin{subequations}
	\begin{align}
	\label{eq:population_coherence1}
	p_i(t) &= \bra{E_i} \rho(t) \ket{E_i},\\
	C_{ij}(t) &= \bra{E_i} \rho(t)\ket{E_j} \quad\mathrm{for~} i \neq j,\label{eq:population_coherence2}
	\end{align} 
\end{subequations}
where $\ket{E_i}$ is the eigenstate of $H$ corresponding to energy $E_i$ (for simplicity we consider non-degenerate $H$). When $\rho(0)$ and $\rho(t_f)$ are diagonal (i.e., $C(0)=C(t_f)=0$), the conditions for Eq.~\eqref{eq:tp} are the well-known thermomajorization constraints of Ref.~\cite{horodecki2013fundamental}.\footnote{In fact, the case $C(0)= 0$, $C(t_f) \neq0$ is trivial, since these transformations are always forbidden.} The problem for general states was formally solved by the remarkable work in Ref.~\cite{gour2018quantum}, where an extremely complex but complete set of entropy conditions was derived. 

One can readily see, e.g. by taking a thermal operation with $E$ being a small environment, that the resource theory and the standard framework apply to different regimes. For a concrete and relevant example, consider a single qubit system and take the environment defining the thermal operation to be given by a single bosonic oscillator in resonance with it. Take the energy-preserving unitary interaction \mbox{$U(t) = \exp(-it H_{\rm int})$} with the interaction Hamiltonian from the Jaynes-Cummings model:
\begin{equation}
	H_{\rm int} = g(\ketbra{1}{0}\otimes a + \ketbra{0}{1}\otimes a^\dag),
\end{equation}
with $a^\dag$ and $a$ denoting the bosonic creation and annihilation operators, and $g$ being a coupling constant. No Markovian master equation like Eq.~\eqref{eq:master} can be derived. In fact, the dynamics is fully solvable, so one can compute $\rho(t)$ and verify that the entropy production relation from Eq.~\eqref{eq:entropy_prod} is violated due to non-Markovian effects. The resource theory approach, in allowing fine control over memory effects and system-bath correlations, does not incorporate the Markovianity condition of standard thermalization processes. 

%------------------------------------------------------------------
% SEC. II.A.3
%------------------------------------------------------------------

\subsubsection{Potential and need for generalizations of the MTP model}

The scenarios that fall beyond the scope of MTPs are most prominently those involving an explicit time-dependence. Even then, however, our framework can still be useful, e.g., $4$-stroke engines have separate thermalization and unitary strokes, and our framework can be used to study the former steps (see the accompanying paper~\cite{korzekwa2022optimizing} where this strategy is explicitly used to study a cooling problem). Importantly, such engines with separate thermalization and unitary steps are thermodynamically universal in the weak coupling regime~\cite{uzdin2015equivalence}. Thus, in this regime, our framework allows one to explore the full potential of quantum heat engines.

However, if one wants to explicitly study thermalization dynamics with an external driving field, then using the MTP model may become problematic. Suppose that one simply allows in Eq.~\eqref{eq:master} an arbitrary time-dependence in $H$, even adding the extra constraint that zero net work is done to the system in (almost) any run. Then, the results of Ref.~\cite{perry2018sufficient} imply that these extended controls simulate arbitrary energy-preserving interactions with a thermal environment of any dimensionality, including generic non-Markovian effects (see also Refs.~\cite{vom2019reachability,schulte2020exploring} for related results). In other words, extending the controls in this way, we recover the standard resource-theoretic approach based on thermal processes discussed above, which \emph{does not} include the relevant constraints of the Markovian master equation approach to thermodynamics. 

Hence, one needs to be very careful in how time-dependence is introduced in the framework. If we are striving to relate the resource theory framework to the standard theory of open quantum systems via Markovian master equations, as we are doing here, we need a more restricted setting. We leave open the question of how to formally generalize MTP further (say, introducing slow driving) without dropping the core Markovianity constraint that we want to impose on the dynamics. The dialogue between resource theories and open quantum system dynamics initiated by this work is far from over.
 
%------------------------------------------------------------------
% SEC. II.B
%------------------------------------------------------------------

\subsection{Insufficiency of current approaches}
\label{sec:insufficiency}

Quantum thermodynamics aims at deriving laws holding independently of the particular dynamics. In other words, based on minimal assumptions (such as assuming that the environment the system interacts with is thermal), one wants to constrain possible state transformations of the system. In what follows we discuss how both the resource-theoretic and the standard approaches fare in this regard, and highlight some important limitations that we wish to overcome with the present contribution. We will consider the stereotypical situation of a system put into a weak thermal contact with a thermal bath at inverse temperature $\beta$. The crucial question is: given a quantum system initially described by a known energy distribution $\v{p}(0)$, what general conditions determine the possible  $\v{p}(t)$ achievable at some later time $t$?

%------------------------------------------------------------------
% SEC. II.B.1
%------------------------------------------------------------------

\subsubsection{Insufficiency of the traditional approach}

\label{sec:insufficiency_standard}

The second law in the form of Eq.~\eqref{eq:entropy_prod} provides a set of conditions that can help us answer the question posed. In fact, the second law provides a functional on the set of states that must be monotonically non-decreasing along the dynamics. To explain its usefulness and limitations, let us focus on a simple example of an incoherent three-level system, i.e., with $d=3$ and $C(0) = 0$. Due to the covariance of Markovian thermal processes (property~\ref{p2}), we necessarily have $C(t) = 0$ for all $t> 0$.\footnote{This is because $C(0)= 0$ is equivalent to $\mathcal{H}(\rho(0)) = 0$, and~\ref{p2} then implies $\mathcal{H}(\mathcal{L}_t(\rho(0))) = 0$. Thus, the state at time $\delta t$ is stationary, $C(\delta t) = 0$, and the argument extends to all $t>0$.} Thus, the entropy $\Sigma$ reads 
\begin{equation}
	\label{eq:entropyproduction1diagonal}
	\Sigma(t) = -\sum_{i=1}^3 p_i(t) (\log p_i(t) - \log \gamma_i),
\end{equation}
where $\v{\gamma}$ denotes the vector of thermal populations. We plot the values of $\Sigma$ in the simplex of all 3-dimensional probability distributions in Fig.~\hyperref[fig:f1levelsets]{1a}, for one particular choice of $\v{\gamma}$. Then, we can verify whether there exists a continuous path $\v{p}(t)$ connecting a given initial $\v{p}(0)$ with the final $\v{p}(t_f)$, with constantly non-decreasing $\Sigma(t)$ along the path (as required by the second law). In particular, for two states with equal entropy $\Sigma$ there is a unique isoentropic path connecting them, and the second law in the form of positive entropy production does not forbid such a transition. 

\begin{figure}
	\begin{minipage}{.5\columnwidth}
			\hspace{-0.5cm}
		\includegraphics[width=1.06\columnwidth]{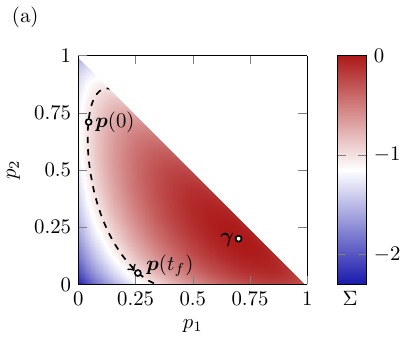}
	\end{minipage}%
	\begin{minipage}{.5\columnwidth}
		\includegraphics[width=1.06\columnwidth]{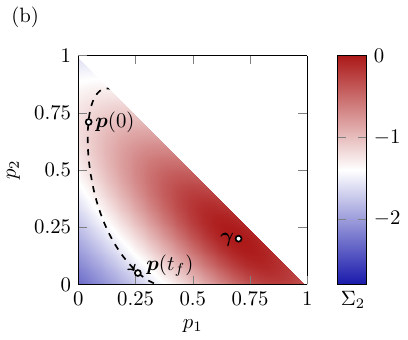}
	\end{minipage}
\caption{\textbf{Thermodynamic entropy landscape for a three-level incoherent system.} The state of an incoherent system is represented by a vector of populations \mbox{$\v{p}=(p_1,p_2,1-p_1-p_2)$}, and the chosen thermal state is \mbox{$\v{\gamma}=(0.7,0.2,0.1)$}. (a)~Entropy functional $\Sigma$. The black dashed trajectory is the path of constant $\Sigma$. The constraint of positive entropy production cannot exclude that the thermodynamic transition from $\v{p}(0)$ to $\v{p}(t)$ is allowed. (b)~Entropy functional $\Sigma_2$. The black dashed trajectory is the path of constant $\Sigma$. The constraint of positive collisional entropy production implies that one \emph{cannot} transform $\v{p}(0)$ into $\v{p}(t)$.}
\label{fig:f1levelsets}
\end{figure}

However, one can prove that for most pairs of isoentropic states there is no Markovian thermal process that can map between them, independently of the chosen controls and details of the environment. What the standard second law is missing is that there exist \emph{generalized (non-equilibrium) entropy production} relations that must be satisfied for a thermodynamic transformation to be allowed. As we will discuss later in Sec.~\ref{sec:complete_set}, one of them is the non-negativity of the \emph{collisional entropy production} (with respect to $\v{\gamma}$):
\begin{equation}
	\label{eq:entropy2_prod}
	\frac{d \Sigma_2(t)}{dt} := -\frac{d}{dt} S_2(\v{p}(t)\|\v{\gamma}) \geq 0, 
\end{equation}  
where 
\begin{equation}
	S_2(\v{p}\|\v{\gamma}) := \log \left(\sum_{i=1}^d \frac{p_i^2}{\gamma_i}\right),
\end{equation}
is the relative R\'{e}nyi entropy of order two~\cite{renyi1961measures}. Eq.~\eqref{eq:entropy2_prod} provides an \emph{independent} monotonically non-decreasing functional for the thermodynamic system. As before, we plot the values of $\Sigma_2$ in Fig.~\hyperref[fig:f1levelsets]{1b} for the same choice of $\v{\gamma}$, so that we can clearly see that the path which kept $\Sigma$ constant is decreasing~$\Sigma_2$. Decreasing the collisional entropy (relative to $\v{\gamma}$) is forbidden, hence the transition is impossible under \emph{any} Markovian thermal process. This rules out, by means of an explicit counterexample, that the standard entropy production functionals (often referred to as the second law of thermodynamics in the literature) faithfully characterizes thermalization out of equilibrium.

%------------------------------------------------------------------
% SEC. II.B.2
%------------------------------------------------------------------

\subsubsection{Insufficiency of the resource-theoretic approach}
\label{sec:insufficiency_resource}

If the standard entropy production constraint is insufficient, the information theory-minded reader could wonder whether the results of the resource theory approach can come to the rescue. Here we will argue why these tools, in their present form at least, are too weak to capture the relevant limitations of quantum thermodynamics in the Markovian regime. 

For an explicit example, consider a two-level incoherent system, i.e., with $d=2$ and $C(0)= 0$. Again, due to covariance property~\ref{p2}, we have $C(t)=0$. Let us choose
\begin{equation}
	\label{eq:qbeta}
	\v{p}(t_f) = \left[\left(1-\frac{\gamma_2}{\gamma_1}\right)p_1 + p_2, \frac{\gamma_2}{\gamma_1} p_1\right].
\end{equation} 
The resource theory then imposes the following family of constraints, known as the \emph{second laws of thermodynamics}~\cite{brandao2013second}:
\begin{equation}
	\label{eq:secondlaws}
	S_\alpha(\v{p}(0)\|\v{\gamma}) \geq S_\alpha(\v{p}(t_f)\|\v{\gamma}), \quad \forall \alpha \in \mathbb{R},  
\end{equation}
where 
\begin{equation}
	S_\alpha(\v{p}\| \v{\gamma}) := \frac{{\rm sgn(\alpha)}}{\alpha-1} \log \left(\sum_{i=1}^d p_i^\alpha \gamma_i^{1-\alpha}\right)
\end{equation}
is the R\'{e}nyi relative entropy of order $\alpha$~\cite{renyi1961measures}. One can show that, for $\v{p}(t_f)$ from Eq.~\eqref{eq:qbeta}, \emph{all} these constraints are satisfied.\footnote{This follows immediately from the fact that $\v{p}(t_f)$ can be obtained from $\v{p}(0)$ by applying the stochastic matrix $G$ with $G_{12} = 1$, $G_{21} = \gamma_2/\gamma_1$. Since $G \v{\gamma} = \v{\gamma}$, Eq.~\eqref{eq:secondlaws} follows, see e.g. Sec.~II.C.2 of Ref.~\cite{lostaglio2019introductory}.} 

Nevertheless, one can prove that $\v{p}(t_f)$ cannot be thermodynamically accessed from $\v{p}(0)$ by a Markovian thermal process. As we shall see in Sec.~\ref{sec:complete_set}, for every Markovian thermal process the following generalized entropy production equations must hold:
\begin{equation}  
	\label{eq:markovian_secondlaws}
	\frac{d \Sigma_\alpha(t)}{dt} := -\frac{d}{dt} S_\alpha(\v{p}(t)\| \v{\gamma})\geq 0, \quad \forall \alpha \in \mathbb{R},  
\end{equation}
which, for $\alpha=1$, recover the standard entropy production inequality from Eq.~\eqref{eq:entropy_prod}. Now, for $d=2$, every dynamical trajectory connecting $\v{p}(0)$ with $\v{p}(t_f)$ can be simply parametrized as 
\begin{equation}
	\label{eq:p_parametrization}
	\v{p}(\lambda) = (1-\lambda)\v{p}(0)+ \lambda \v{p}(t_f).
\end{equation}
In Fig.~\ref{fig:falphalevelsets} we plot $\Sigma_\alpha$ as a function of $\alpha$ and the trajectory parameter $\lambda$ for a particular choice of $\v{p}(0)$ and $\v{\gamma}$. One can clearly see in the plotted range of $\alpha$ that $\Sigma_\alpha$ at $\lambda = 0$ is smaller than at $\lambda =1$, but for each $\alpha$ there exists an intermediate point $\lambda_{*}(\alpha)\in(0,1)$ at which $\Sigma_\alpha$ starts to decrease. Hence, the $\alpha$-entropy production inequalities from Eq.~\eqref{eq:markovian_secondlaws} are violated at some intermediate time, even if the second laws of Eq.~\eqref{eq:secondlaws} are all satisfied. This is why no Markovian thermal process mapping $\v{p}(0)$ to $\v{p}(t_f)$ exists, a fact that cannot be captured by the end-points condition of the resource theory approach. The conceptual issue is clear: the resource theory approach only considers discrete transformations and does not involve notions such as that of a continuously generated process constantly producing entropy along its path.

Furthermore, it should be clear that even Eq.~\eqref{eq:markovian_secondlaws} (a strengthening of the second laws of Ref.~\cite{brandao2013second}) is not entirely satisfactory. In fact, beyond the simplest cases, one cannot exhaustively check an infinite set of inequalities along arbitrary trajectories with fixed end-points.

\begin{figure}[t]
	\centering
	\includegraphics[width=0.6\columnwidth]{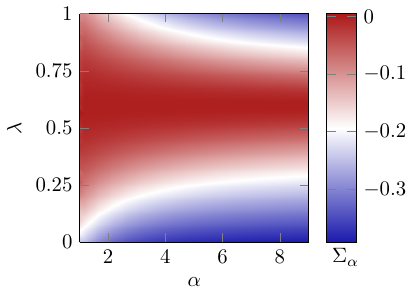}
	\caption{\textbf{Thermodynamic $\alpha$-entropy for a two-level system.} 		
		Value of the relative entropy functional $\Sigma_\alpha$ along the trajectory specified in Eq.~\eqref{eq:p_parametrization} connecting $\v{p}(0)= (0.1,0.9)$ with $\v{p}(t_f)$ given by Eq.~\eqref{eq:qbeta}. The chosen thermal state is $\v{\gamma} =(0.6,0.4)$. Note that, for every $\alpha$, the path connecting $\lambda = 0$ to $\lambda =1$, necessarily passes through a region in which entropy $\Sigma_\alpha$ decreases. Hence, no Markovian thermal process can transform $\v{p}(0)$ into $\v{p}(t_f)$.}
	\label{fig:falphalevelsets}
\end{figure}

%------------------------------------------------------------------
% SEC. III
%------------------------------------------------------------------

\section{Proposed hybrid approach}
\label{sec:approach}

To address the issues highlighted in the previous section, we propose to study the set of Markovian thermal processes using information theory tools, incorporating from the get-go constraints that are commonplace in most quantum thermodynamic settings, such as the presence of a large heat bath, weak coupling and Markovianity. Our purpose is to leverage the information theory tools to complement the toolkit of the master equation formalism. 

The central question investigated in this paper is: what final states $\rho(t_f)$  are accessible from an initial state $\rho(0)$ by means of Markovian thermal processes? When such a process exists transforming $\rho(0)$ into $\rho(t_f)$ we will write
\begin{equation}
	\label{eq:mtp}
	\rho(0) \stackrel{ \textrm{MTP}}{\longmapsto} \rho(t_f).
\end{equation}
Our main contribution is to find a complete set of conditions to answer this question when $\rho(t_f)$ is block-diagonal in the energy basis. Due to property~\ref{p2}, the problem is reduced to the one involving energy distributions (`populations')
\begin{equation}
	\label{eq:mtpdiagonal}
	\v{p}(0) \stackrel{ \textrm{MTP}}{\longmapsto} 	\v{p}(t_f). 
\end{equation}
 
When approached directly, this may look like an extremely complex control problem. As we have seen, it is not enough to find a trajectory $\v{p}(t)$ connecting $\v{p}(0)$ to $\v{p}(t_f)$ involving irreversible entropy production (which in itself is a hard task), because such a trajectory does not guarantee that a master equation achieving the desired transformation exists. A numerical brute force approach is also unfeasible, since it involves the exploration of a very high dimensional space of control parameters. Ultimately, this is related to the fact that characterizing what dynamics can be realized by a Markovian master equation is an extremely challenging problem even classically: this is known as the \emph{embeddability problem}~\cite{elfving1937theorie,davies2010embeddable}. Despite having been studied for decades in the mathematics literature, general analytic solutions are not known beyond the simplest $d=2$ and $d=3$ case~\cite{kingman1962imbedding,runnenburg1962elfving,goodman1970intrinsic,carette1995characterizations}. 

Lacking explicit characterizations, we will follow a different strategy. It is crucial to highlight that the solution will satisfy two desiderata:
\begin{enumerate}[label=(D\arabic*)]
	\item\label{d1}\textbf{Finite verifiability.} One should be able to verify in a finite number of steps whether \mbox{$\v{p}(0) \stackrel{ \textrm{MTP}}{\longmapsto} \v{p}(t_f)$} holds for any given initial and final states.
	\item\label{d2} \textbf{Constructability.} Whenever \mbox{$\v{p}(0) \stackrel{ \textrm{MTP}}{\longmapsto} \v{p}(t_f)$} holds, one should be able to explicitly construct a Markovian thermal process realizing this transition through a sequence of elementary controls.  
\end{enumerate}
These are central requirements for the applicability of the framework and in typical resource theory approaches these are not both satisfied. 

%------------------------------------------------------------------
% SEC. IV
%------------------------------------------------------------------

\section{Continuous thermomajorization}
\label{sec:thermomajo}

Here we introduce the main technical tool to solve the problem at hand: a generalization of \emph{thermomajorization}. We start with a summary of well-known results.

\subsection{Recap: majorization and thermomajorization}

\emph{Majorization} is an ubiquitous relation between pairs of vectors that finds applications in fields ranging from mathematics and economy to information theory and quantum physics~\cite{marshall2010inequalities, nielsen2002introduction}. Given two probability distributions, $\v{p}$ and $\v{q}$, we say that $\v{p}$ \emph{majorizes} $\v{q}$, denoted $\v{p} \succ \v{q}$, if 
\begin{align}
	\sum_{i=1}^j p_i^\downarrow\geq \sum_{i=1}^j q_i^\downarrow \quad \textrm{for} \; j=1,\dots, d,
\end{align}
where $\v{x}^{\downarrow}$ denotes the vector $\v{x}$ sorted in a non-increasing order. The partial ordering of probability vectors induced by majorization can be seen as formalizing the measure of disorder relative to the uniform distribution~\mbox{$\v{\eta}:=(1/d,\dots,1/d)$}: for a fixed dimension $d$, sharp distributions majorize all other distributions, and all distributions majorize the uniform distribution~$\v{\eta}$. Furthermore, if $\v{p} \succ \v{q}$, then the Shannon entropy of $\v{p}$ is smaller than that of~$\v{q}$. The same holds for a whole class of entropy functionals known as Schur-concave functions~\cite{marshall2010inequalities} (including all R\'enyi entropies~\cite{renyi1961measures}).

However, just like with the entropy production in Eq.~\eqref{eq:entropy_prod}, it is convenient to extend the notion of majorization to thermomajorization~\cite{ruch1978mixing, horodecki2013fundamental} (or majorization relative to $\v{\gamma}$, or $\v{\gamma}$-majorization), so that disorder is measured relative to a generic non-uniform equilibrium distribution $\v{\gamma}$. To do so, first denote by $\v{\pi}(\v{p})$ the reordering of $\{1,\dots, d\}$ that sorts $p_i/\gamma_i$ in a non-increasing order,
\begin{equation}
	\label{eq:betaorder}
	\frac{p_{\pi_i(\v{p})}}{\gamma_{\pi_i(\v{p})}} \geq \frac{p_{\pi_{i+1}(\v{p})}}{\gamma_{\pi_{i+1}(\v{p})}} \quad \textrm{for} \; i=1,\dots, d-1.
\end{equation}
This is called the $\v{\gamma}$-\emph{ordering} or \emph{thermomajorization ordering} of $\v{p}$.
Next, we need to introduce the notion of \emph{Lorenz curve}. The Lorenz curve of $\v{p}$ relative to $\v{\gamma}$ (also called thermomajorization curve in Ref.~\cite{horodecki2013fundamental}) is a piecewise linear, concave curve on a plane that connects the points $\v{l}^{(j)}$ given by
\begin{align}
	\v{l}^{(j)} = \left(\sum_{i=1}^j \gamma_{\pi_i(\v{p})}, \sum_{i=1}^j p_{\pi_i(\v{p})}\right)
\end{align}
for $j\in\{1,\dots,d\}$, where $\v{l}^{(0)}:= (0,0)$. Then, $\v{p}$ is said to thermomajorize $\v{q}$ (relative to $\v{\gamma}$), denoted \mbox{$\v{p} \succ_{\v{\gamma}} \v{q}$}, when the Lorenz curve of $\v{p}$ is never below that of $\v{q}$.\footnote{If we denote the height of the Lorenz curve at point $x \in [0,1]$ by $L_x(\v{p}\|\v{\gamma})$, this can also be written as  $L_{x_j}(\v{p}\|\v{\gamma}) \geq L_{x_j}(\v{q}\|\v{\gamma})$ for $j\in\{1,\dots d\}$, where $x_j= \sum_{i=1}^j \v{\gamma}_{\pi_i(\v{q})}$~\cite{alhambra2016fluctuating}.} Importantly, in the case of uniform equilibrium distributions, $\v{\gamma} = \v{\eta}$, thermomajorization  reduces to majorization. For a fixed dimension, the sharp distribution with largest energy, $(0, \dots , 0, 1)$, thermomajorizes every other distribution, and every distribution thermomajorizes~$\v{\gamma}$. Furthermore, if $\v{p} \succ_{\v{\gamma}} \v{q}$, then the relative entropy $S(\v{p}\| \v{\gamma})$ is larger than $S(\v{q}\| \v{\gamma})$. The same holds for a general class of relative entropy functionals called thermodynamic Schur-concave functions in Ref.~\cite{lostaglio2019introductory} (including all $\alpha$-relative entropies $S_\alpha(\cdot\|\v{\gamma})$ in Eq.~\eqref{eq:secondlaws}). 
 
 \begin{figure}[t]
 	\centering
 	\includegraphics[width=0.6\columnwidth]{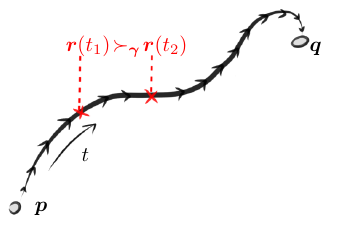}
 	\caption{\label{fig:defcontinuousmajorization}\textbf{Continuous thermomajorization}. The continuous thermomajorization relation $\v{p}\ggcurly_{\v{\gamma}} \v{q}$ holds if and only if there is a continuous path of probability distributions $\v{r}(t)$ connecting $\v{p}$ and $\v{q}$ such that $\v{r}(t_1)\succ_{\v{\gamma}}\v{r}(t_2)$ whenever $t_1\leq t_2$.   }
\end{figure}
The crucial property of thermomajorization as a partial ordering of probability vectors is that it characterizes transformations under thermal processes~\cite{horodecki2013fundamental, lostaglio2019introductory}:
\begin{equation}
	\label{eq:thermomajorization}
\v{p}(0) \stackrel{ \textrm{TP}}{\longmapsto} \v{p}(t_f)\quad \Leftrightarrow\quad 	\v{p}(0) \succ_{\v{\gamma}} \v{p}(t_f).
\end{equation}
Note that Eq.~\eqref{eq:thermomajorization} is satisfied for every Markovian thermal process, since these are a subset of thermal processes. However, the problem we pointed out in Sec.~\ref{sec:insufficiency} remains: these end-point conditions do not capture the existence of a continuous process generated by a Markovian master equation. 
 
 \subsection{Continuous thermomajorization}
 
 We introduce the following strengthening of thermomajorization that we illustrate in Fig.~\ref{fig:defcontinuousmajorization}. 
 
\begin{defn}[Continuous thermomajorization]
	\label{def:markov_majo}
 	A distribution $\v{p}$ \emph{continuously thermomajorizes} $\v{q}$ (or continuously majorizes $\v{q}$ relative to $\v{\gamma}$), denoted $\v{p} \ggcurly_{\v{\gamma}} \v{q}$, if there exists a continuous path of probability distributions $\v{r}(t)$ for $t\in[0,t_f)$ such that
 	\begin{enumerate}
 		\item $\v{r}(0)=\v{p}$,
 		\item $\forall~ t_1,t_2\in[0,t_f):\quad t_1 \leq t_2 \Rightarrow\v{r}(t_1)\succ_{\v{\gamma}} \v{r}(t_2)$,
 		\item $\v{r}(t_f)=\v{q}$.
 	\end{enumerate}
We call $\v{r}(t)$ a \emph{thermomajorizing trajectory} from $\v{p}$ to $\v{q}$.
\end{defn}
 
Note that in the particular case of a uniform fixed point, $\v{\gamma}=\v{\eta}$, the above definition corresponds to a continuous version of standard majorization, denoted by $\ggcurly$ in Ref.~\cite{zylka1985note}. In fact, the notion of continuous majorization has a decades-long history and appears in a variety of research fields from thermodynamics and order theory~\cite{zylka1990accessibility,alberti2008order}, through plasma physics~\cite{hay2015maximal,hay2017extreme}, to social sciences~\cite{thon2004dalton}. Moreover, this notion was employed and studied in more detail in Ref.~\cite{zylka1985note}, where it was inspired by a model of heat transport along ideal conducting wires between $d$ objects with different temperatures. Here, we extend these technical considerations to continuous thermomajorization, which is necessary to capture finite temperature thermalizations. We will also highlight the significance of this notion as the right generalization of the concept of entropy production.
  
Our first main result is to show that the notion of continuous thermomajorization correctly encapsulates \emph{all} the relevant constraints of Markovian thermal processes on population dynamics.
 
\begin{thm}[Second law on populations]
	\label{thm:main}
 	\mbox{$\v{p}(0) \stackrel{ \textrm{MTP}}{\longmapsto} 	\v{p}(t_f)$} if and only if
 	\begin{equation}
 		\label{eq:continuousmajorization}
 		\v{p}(0) \ggcurly_{\v{\gamma}} \v{p}(t_f).
 	\end{equation}
\end{thm}

The proof of the above theorem can be found in Appendix~\ref{app:generaltheorems}. As a consequence, continuous thermomajorization gives a complete (exhaustive) set of constraints for the evolution of populations in the standard (Markovian) master equations approach to quantum thermodynamics. 

It is also worth highlighting that continuous thermomajorization -- which characterizes Markovian processes -- and thermomajorization -- which characterize general non-Markovian processes -- coincide when initial and final states have the same $\v{\gamma}$-ordering (see Corollary~\ref{corol:orderingscoincide} in Appendix~\ref{app:generaltheorems} for details):
\begin{equation}
	\textrm{If} \; \; \v{\pi}(\v{p}) = \v{\pi}(\v{q}), \quad \v{p} \succ_{\v{\gamma}} \v{q} \Leftrightarrow \v{p} \ggcurly_{\v{\gamma}} \v{q}.
\end{equation}
In other words, \emph{all the complications with Markovianity (or advantages from non-Markovianity) arise from crossing the boundary between one $\v{\gamma}$-ordering and another}. This observation will play a crucial role later, but more broadly it is noteworthy for the study of the role of memory effects in stochastic processes with a given fixed point.

%------------------------------------------------------------------
% SEC. V
%------------------------------------------------------------------

\section{A complete set of entropy production relations}
\label{sec:complete_set}
  
We now show how the continuous thermomajorization condition of Theorem~\ref{thm:main} subsumes (and greatly strengthens) the standard positive entropy production condition from Eq.~\eqref{eq:entropy_prod}. Employing Theorem~\ref{thm:main}, one can translate known results from the theory of majorization into entropic inequalities. In other words, one can construct families of functionals that must be monotonically non-decreasing during the Markovian evolution of the system along the path $\rho(t)$ with populations $\v{p}(t)$. For example, for any well-behaved convex function $h: \mathbb{R} \rightarrow \mathbb{R}$, the $h$-divergence defined by
\begin{equation}
	\label{eq:h_divergence}
	\Sigma_h(t) = - \sum_{i=1}^d \gamma_i h\left( \frac{p_i(t)}{\gamma_i} \right),
\end{equation}
must be monotonically non-decreasing
\begin{equation}
	\label{eq:entro_prod}
	\frac{d\Sigma_h(t)}{dt} \geq 0.
\end{equation}
\begin{proof}
To see that Eq.~\eqref{eq:entro_prod} holds, note that, by Theorem~\ref{thm:main}, \mbox{$\v{p}(0) \ggcurly_{\v{\gamma}} \v{p}(t_f)$}. Hence, for any \mbox{$t \in [0, t_f)$} and $\delta >0$, $\v{p}(t) \succ_{\v{\gamma}} \v{p}(t+ \delta)$. The known results on thermomajorization (see Theorem~\ref{thm:thermomajorizationprevious} in Appendix~\ref{app:generaltheorems}) then tell us that there exists a stochastic matrix $T$ such that 
	\begin{equation}
		T \v{p}(t) = \v{p}(t+\delta), \quad T \v{\gamma} = \v{\gamma}.
	\end{equation}
	Thus,
	\begin{align}
		\Sigma_h(\v{p}(t+\delta)) & =  - \sum_{i=1}^d \gamma_i h\left(\sum_{j=1}^d T_{ij} \frac{p_j(t)}{\gamma_i} \right)  \nonumber\\
		& = - \sum_{i=1}^d \gamma_i h \left( \sum_{j=1}^d \left[ T_{ij} \frac{\gamma_j}{\gamma_i} \right] \frac{p_j(t)}{\gamma_j} \right)   \nonumber\\
		& \geq -\sum_{i,j=1}^d \gamma_i \left[ T_{ij} \frac{\gamma_j}{\gamma_i} \right] h\left(\frac{p_j(t)}{\gamma_j} \right)\nonumber\\		
		&\geq - \sum_j \gamma_j h\left( \frac{p_j(t)}{\gamma_j} \right) = \Sigma_h (\v{p}(t)),
		\label{eq:monotonicityh}
	\end{align}  
	where we used the convexity of $h$ and then the stochasticity of $T$ (i.e., $\sum_i T_{ij} =1$). We note Eq.~\eqref{eq:monotonicityh} could have also be inferred from Ref.~\cite{marshall2010inequalities}, Proposition 14.B.3. Since the above holds for every $\delta >0$, the result follows.
\end{proof}

For each choice of $h$, the above qualifies as a valid \emph{generalized entropy production} inequality. Restrictions on the thermodynamically admissible paths can be obtained by studying their level sets within the $d$-dimensional probability simplex (recall Figs.~\hyperref[fig:f1levelsets]{1a}-\hyperref[fig:f1levelsets]{1b}) and constructing the corresponding ``thermodynamic trees'', as detailed in Ref.~\cite{gorban2013thermodynamic} for a special choice of $\Sigma_h$. In the accompanying paper~\cite{korzekwa2022optimizing}, we detail how these encompass and strengthen several well-known relations in the literature, including the standard entropy production relation of Eq.~\eqref{eq:standard_entropy_prod}, the diagonal entropy production~\cite{santos2019role}, the second laws of Ref.~\cite{brandao2013second}, the Tsallis entropies well-known in non-extensive statistical mechanics and information theory~\cite{tsallis1988possible, abe2000axioms, tsallis2009introduction, mariz1992irreversible} and the `vacancy'~\cite{wilming2017third}, which was found to play a crucial role in low temperature thermodynamics.

As we can see, one can easily generate a huge variety of entropic inequalities, which helps to see different results as part of a unified framework. At the same time, a natural question arises: Is there a family of entropic conditions that implies \emph{all others}? Our second main result, which follows from Theorem~\ref{thm:main}, answers this question in the affirmative and can be interpreted as a sort of exhaustive $H$-theorem.
\begin{cor}[Exhaustive $H$-type theorem]
	\label{thm:gen_entropy_prod}
	\mbox{$\v{p}(0) \stackrel{ \textrm{MTP}}{\longmapsto} \v{p}(t_f)$} if and only if there exists a continuous path $\v{p}(t)$ for $t\in[0,t_f]$ such that $\Sigma_a(t)$ is monotonically not decreasing in $t$ for all $ a \in [0, 1]$, where
 \begin{equation}
	\label{eq:generalizedentropicrelations}
	\Sigma_a(t) := -\sum_{i=1}^d \left|p_i(t) - a \frac{\gamma_i}{\gamma_d}\right|.
\end{equation}
\end{cor} 
	\begin{proof}
		First assume that the evolution of populations $\v{p}(t)$ is generated by a Markovian thermal process. Then, from Theorem~\ref{thm:main}, we know that for every $\epsilon>0$ we have $\v{p}(t) \succ_{\v{\gamma}} \v{p}(t+\epsilon)$. This is equivalent to (see Chap. 14, Proposition B.4~\cite{marshall2010inequalities}):
		\begin{equation}
			\sum_i |p_i(t) - a \gamma_i| \geq \sum_i |p_i(t+\epsilon) - a \gamma_i|, \quad \forall a \geq 0. 
		\end{equation}
		Since $\epsilon>0$ can be made arbitrarily small, this is the condition that the functionals 	$\sum_i |p_i(t) - a \gamma_i|$ are monotonically non-increasing in $t$ for every $a\geq 0$. These include in particular the monotonicity of the functionals $\Sigma_a(t)$.
		
		Conversely, suppose that the evolution of populations $\v{p}(t)$ is such that the functionals $\Sigma_a(t)$ are monotonically non-decreasing for all $a\in[0,1]$. Note that for $a>1$ one has $a \gamma_i/\gamma_d >1$ for every $i$, and hence $\sum_i |p_i(t) - a \frac{\gamma_i}{\gamma_d}| =\frac{a}{\gamma_d}-1$ independently of the value of $\v{p}(t)$. Thus, we can trivially extend the monotonicity property of $\Sigma_a(t)$ to all $a \geq 0$. In fact, by rescaling $a \mapsto a \gamma_d$, we get the equivalent property that $\sum_i |p_i(t) - a \gamma_i|$ is monotonically non-decreasing. As already mentioned above, this is equivalent to $\v{p}(t) \succ_{\v{\gamma}} \v{p}(t+\epsilon)$ for every $t, \epsilon \geq 0$. Recalling Definition~\ref{def:markov_majo}, this means that for every $t$ we have $\v{p}(0) \ggcurly_{\v{\gamma}} \v{p}(t)$. We conclude using Theorem~\ref{thm:main}. 
	\end{proof}
 Once again, we want to emphasize the ``if and only if'' in the statement: $d\Sigma_a(t)/dt \geq 0$ are generalized entropy production inequalities which imply all others. 

%------------------------------------------------------------------
% SEC. VI
%------------------------------------------------------------------

\section{Universal thermodynamic controls}
\label{sec:controls}

We now change the point of view and consider the equally important question of \emph{control}. In other words, we start from an initial state $\v{p}(0)$ and ask how to devise a thermalization process that drives the system to a final target state $\v{p}(t_f)$ at some later time $t_f$. From Theorem~\ref{thm:main} we know that every $\v{q}$ such that $\v{p}(0) \ggcurly_{\v{\gamma}} \v{q}$ can be realized by \emph{some} choice of controls in the class of Markovian master equations of a thermalization process. Such controls, however, may be arbitrarily complex and the control sequence is unknown. In this section we solve the first problem by presenting a set of elementary controls that are sufficient to perform arbitrary Markovian thermalizations; and in the next section we will solve the second problem by presenting an algorithm that returns the explicit sequence that is required. 

One can reasonably conjecture that a much more restricted subclass of physically relevant thermalization processes suffices to grant us the \emph{same} amount of control as the full set of Markovian thermal processes. The reader can be reminded of the notion of a \emph{universal gate set} in quantum computing, where one seeks a minimal set of unitary operations that allows one to approximate arbitrarily well the transformations achievable by arbitrary unitaries~\cite{nielsen2010quantum}. In the same fashion, we ask here about a set of \emph{universal thermalization controls}.

Following this intuition, Ref.~\cite{lostaglio2018elementary} asked whether every transformation achievable by thermal processes can be achieved by sequentially coupling only two energy levels of the system to the environment at once, dubbed an `elementary thermal operation'. Somewhat surprisingly, this question was answered in the negative~\cite{lostaglio2018elementary}. In fact, one needs to couple the environment simultaneously to all $d$ energy levels~\cite{mazurek2018decomposability}, or grant full control of the system's and an auxiliary thermal qubit's energy spectra~\cite{perry2018sufficient}. 

Remarkably, however, we are not aware that the same question was tackled in the standard setup of quantum thermodynamics, where the question is to find controls as powerful as the most general Markovian master equation of a thermalization process. In this context, a distinguished candidate for a universal set of thermal controls is given by \emph{two-level partial thermalizations}, which we will also simply call \emph{elementary thermalizations}. These are a set of thermalizations of both practical and formal interest. Each of them acts only on two energy levels $(i,j)$ and is represented by an extremely simple reset Markovian master equation
\begin{subequations}
	\begin{align*}
		\frac{dp_i}{dt}&=\frac{1}{\tau}\left(  \frac{\gamma_i}{\gamma_i+\gamma_j}(p_i+p_j)-p_i \right), \quad 	\frac{dp_j}{dt} = - 	\frac{dp_i}{dt}.
	\end{align*}
\end{subequations}
which describes an exponential relaxation to equilibrium:
\begin{equation}
	\v{p}^{i,j}(t) = e^{-t/\tau} \v{p}^{i,j}(0) + N_{ij}(0) (1-e^{-t/\tau}) \v{\gamma}^{i,j}.
\end{equation} 
Above, $\v{x}^{i,j}(t):=(x_i(t),x_j(t))$ and $N_{ij} = p_i(0) + p_j(0)$. Formally, this can be represented by a matrix equation
\begin{equation}
	\v{p}^{i,j}(t) = T^{i,j}(\lambda_t) \v{p}^{i,j}(0)
\end{equation}
with $\lambda_t = 1- e^{-t/\tau}$ and
\begin{equation}
	\label{eq:elementarythermalization}
		T^{i,j}(\lambda) =
		\begin{bmatrix}
			(1-\lambda) + \frac{  \lambda \gamma_i}{\gamma_i+\gamma_j} &	\lambda \frac{\gamma_i}{\gamma_i+\gamma_j} \\
			\lambda \frac{\gamma_j}{\gamma_i+\gamma_j} & 	(1-\lambda) +  \frac{\lambda \gamma_i}{\gamma_i+\gamma_j}
		\end{bmatrix}.
\end{equation}

These transformations stand out for their formal simplicity -- they are the stochastic processes with thermal fixed point on two states that can be realized by a Markovian master equation (as one can check directly using so-called embeddability conditions~\cite{davies2010embeddable}). But they also arise naturally in rather diverse approaches to quantum thermodynamics~\cite{davies1974markovian, roga2010davies, scarani2002thermalizing}, where they are often used as building blocks for more complex protocols~\cite{perry2018sufficient, baumer2019imperfect, miller2019work}. Here, we prove that elementary thermalizations are a universal set of thermalization controls:

\begin{thm}[Universality of elementary thermalizations]
	\label{thm:universality}
	\mbox{$\v{p}(0) \stackrel{ \textrm{MTP}}{\longmapsto} 	\v{p}(t_f)$} if and only if there exists a finite sequence of elementary thermalizations such that 
	\begin{equation}
		\label{eq:plttrajectory}
		\v{p}(t_f)=T^{i_f,j_f}(\lambda_f) \dots T^{i_1,j_1}(\lambda_1) \v{p}(0).
	\end{equation}
\end{thm}
For the proof, see Appendix~\ref{app:generaltheorems}. This is a remarkable simplification of the set of controls required to generate the transformations achievable by the most general Markovian thermal process.\footnote{The result also shows that Markovian thermal processes on incoherent states have the same power as Markovian thermal operations, since every elementary thermalization can be easily seen to be a thermal operation as defined in Ref.~\cite{brandao2011resource}.} We remark once more that this simplification does not hold for thermal processes or thermal operations~\cite{lostaglio2018elementary}, and  so constitutes an important difference between the standard and the resource theory frameworks.\footnote{Notable exceptions are given by infinite temperature limit and systems with trivial Hamiltonians, when the thermal state is a maximally mixed state and the thermomajorization relation is replaced by standard majorization. Then, it is known (see, e.g., Theorem~II.1.10 of Ref.~\cite{bhatia2013matrix}) that majorization between $\v{p}$ and $\v{q}$ is equivalent to the existence of a finite sequence of $T$-transforms mapping $\v{p}$ to $\v{q}$, where the $T$-transform is a bistochastic matrix acting non-trivially only on two levels of the system.} Our result proves that coupling at once more than two system energy levels to the environment is only required when we want to reproduce effects arising from strong interactions or small environments, but it is not necessary in the Markovian regime.  

%------------------------------------------------------------------
% SEC. VII
%------------------------------------------------------------------

\section{Second laws in the Markovian regime}
\label{sec:second_laws}

We are now ready to state the main results of this work. First, we will provide a finite set of necessary and sufficient conditions for a given probability distribution~$\v{p}$ to continuously thermomajorize another distribution~$\v{q}$ and so, via Theorem~\ref{thm:main}, for $\v{p} \stackrel{\textrm{MTP}}{\longmapsto} \v{q}$. Second, we will specify a constructive protocol realizing this transition through a sequence of elementary thermalizations. Therefore, our results satisfy desiderata~\ref{d1}~and~\ref{d2}.

%------------------------------------------------------------------
% SEC. VII.A
%------------------------------------------------------------------

\subsection{Finite set of conditions}
\label{sec:finite_set}

In order to state our main result, we will need the concept of a \emph{canonical sequence} of $\v{\gamma}$-orderings, defined as follows. 
\begin{defn}
	\label{def:canon_seq}
A sequence of $\v{\gamma}$-ordering vectors $\{\v{\pi}^k\}$ is canonical when
\begin{enumerate}
	\item $\v{\pi}^k$ and $\v{\pi}^{k+1}$ differ only by a transposition of adjacent elements.
	\item Each $\v{\gamma}$-ordering appears at most once in the sequence.
\end{enumerate}
\end{defn}
	
We then have the following result.
\begin{thm}[Finite second laws conditions]
	\label{thm:finite}
	Given $\v{p}$ and $\v{q}$, enumerate all canonical sequences $\{\v{\pi}^k\}_{k=1}^N$ with $\v{\pi}^1 = \v{\pi}(\v{p})$ and $\v{\pi}^N = \v{\pi}(\v{q})$. For each sequence, construct the state
	\begin{equation}
		\v{f} : = 	\prod_{k=1}^{N-1} T^{i_k,j_k}(1) \v{p},
	\end{equation}
	where $T^{i_k,j_k}(1)$ are full elementary thermalizations (Eq.~\eqref{eq:elementarythermalization} with $\lambda =1$) and the levels $i_k$, $j_k$ are the labels indicating which of the elements of $\v{\pi}^k$ and $\v{\pi}^{k+1}$ differ. Then $\v{p} \ggcurly_{\v{\gamma}} \v{q}$ if and only if for at least one $\v{f}$
	\begin{equation}
		\label{eq:finitesequence}
		\v{f} \succ_{\v{\gamma}} \v{q},
	\end{equation}

\end{thm}

\begin{figure}[t]
	\centering
	\includegraphics[width=\columnwidth]{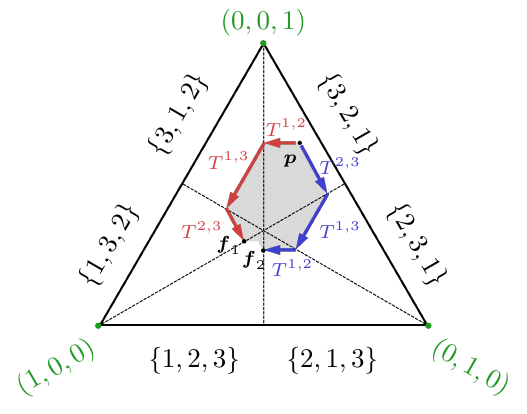}
	\caption{\label{fig:algorithm}\textbf{Verification of continuous thermomajorization for $d=3$.}
		Simplex representing the state space of all 3-dimensional probability distributions with regions of fixed $\v{\gamma}$-orderings indicated by $\{\cdot,\cdot,\cdot\}$. The optimal paths connecting a state $\v{p}(0)$ with $\v{\gamma}$-ordering $\{3,2,1\}$ to states $\v{f}_1$, $\v{f}_2$ of $\v{\gamma}$-ordering $\{1,2,3\}$ are realized by elementary thermalizations $T^{i,j}$ (indicated by red and blue arrows). The set of states with $\v{\gamma}$-ordering $\{1,2,3\}$ achievable from $\v{p}(0)$ by Markovian thermal processes is finally obtained as the union of the set of $\v{q}$ thermomajorized by $\v{f_1}$ and the set of $\v{q}$ thermomajorized by $\v{f}_2$ (for this last construction see, e.g., Appendix~E of Ref.~\cite{lostaglio2018elementary}). Note that the set of achievable states is non-convex. Here the thermal state was chosen to be $\v{\gamma}=[1/3,1/3,1/3]$, which corresponds to the infinite temperature limit.  
	}
\end{figure}
	
The proof of the above theorem can be found in Appendix~\ref{app:thm4}. Before we analyse the complexity of verifying the relation $\v{p} \ggcurly_{\v{\gamma}} \v{q}$ with the use of Theorem~\ref{thm:finite}, let us first illustrate it with an explicit example with $d=3$ and $\v{\gamma} = (1/3, 1/3, 1/3)$ (see Fig.~\ref{fig:algorithm}). In this case, the $\v{\gamma}$-ordering is simply the standard sorting in a non-increasing order and $\succ_{\v{\gamma}}$ coincides with standard majorization $\succ$. Let us choose $\v{p}(0)$ with ordering $\{3,2,1\}$ and $\v{p}(t_f)$ with ordering $\{1,2,3\}$. There are then only two canonical sequences from $\v{\pi}(\v{p}(0))$ to $\v{\pi}(\v{p}(t_f))$: 
	\begin{subequations}
	\begin{align}
		\{3,2,1\} \rightarrow \{2,3,1\} \rightarrow \{2,1,3\} \rightarrow \{1,2,3\}, \\
		\{3,2,1\} \rightarrow \{3,1,2\} \rightarrow \{1,3,2\} \rightarrow \{1,2,3\}.
	\end{align}
	\end{subequations}
	For each sequence we construct a final state:
	\begin{subequations}
	\begin{align}
		\v{f}_1 &= T^{1,2}(1) T^{1,3}(1) T^{2,3}(1) \v{p}(0),\\
		\v{f}_2 &= T^{2,3}(1) T^{1,3}(1) T^{1,2}(1) \v{p}(0).
	\end{align}
	\end{subequations}
	According to Theorem~\ref{thm:finite}, $\v{p} \ggcurly_{{\v{\gamma}}} \v{q}$ if and only if either $\v{f}_1 \succ_{\v{\gamma}} \v{p}(t_f)$ or $\v{f}_2 \succ_{\v{\gamma}} \v{p}(t_f)$, which is simply a set of \mbox{$2(d-1)=4$} inequalities. Employing Theorem~\ref{thm:main}, whenever these inequalities are satisfied there exists a Markovian thermal process mapping $\v{p}(0)$ to $\v{p}(t_f)$.

The above considerations yield a finite set of inequalities and can hence be checked algorithmically, which in turn means that desideratum~\ref{d1} is met. However, how does the approach fare in practice? Taken at face-value, the complexity of the procedure given in Theorem~\ref{thm:finite} grows \emph{extremely} quickly with $d$. To see this, consider a graph whose vertices are all permutations of $\{1,\dots,d\}$, and whose edges connect strings differing by a transposition of adjacent elements. Then, the number of canonical sequences is equal to the number of \emph{simple paths} connecting two points on the graph. In the worst-case scenario, an upper bound on this number scales asymptotically as $O(d^{d!})$. Numerically, we go from $2$ for $d=3$ to over $2000$ for $d=4$.
	
An explicit algorithm based on the above theorem will have to adopt a more clever strategy. The general intuition that we will implement in detail in the next section is the following. The theorem mandates to start from $\v{p}(0)$ and construct all states $\{\v{f}'\}$ corresponding to canonical sequences of length $1$; then, each of those is extended to generate all $\{\v{f}''\}$ states corresponding to canonical sequences of length $2$; and so on, up to the maximum length $d!-1$. Crucially, however, we will terminate all paths that are provably worse than those already constructed. If a canonical sequence leads to a state $\v{f}_1$ with a given $\v{\gamma}$-ordering, and another leads to a state $\v{f}_2\succ_{\v{\gamma}} \v{f}_1$ with the same $\v{\gamma}$-ordering, the former path will be terminated. This is justified by transitivity of thermomajorization: every state achievable by extending further the path containing $\v{f}_1$ can also be achieved by extending the path containing $\v{f}_2$, and can thus be removed from the set of conditions in Eq.~\eqref{eq:finitesequence}. Pictorially, we will construct an algorithm pushing forward all paths at once, while also constantly monitoring which ones can be terminated. 

Intuition suggests that very long paths (of length $O(d!)$) are unlikely to be required, since they involve many more thermalizations than the shortest paths connecting $\v{p}(0)$ to a $\v{f}$ with the same $\v{\gamma}$-ordering as the target $\v{q}$. Hence, we expect the algorithm will terminate sequences before they become too long, leading to large savings compared to the worst-case bound. This is confirmed by numerics: for random initial states and random thermal distributions $\v{\gamma}$ we find that, on average, the longest surviving path to generate every accessible final state has length $L \approx 6.4$ for $d=4$ and $L \approx 11.4$ for $d=5$. These should be compared with the longest possible paths for $d=4$ (23) and $d=5$ (120). In fact, these $L$'s are very close to the length of the shortest sequence connecting two arbitrary points in the above-mentioned graph. This minimal length is equal to $d(d-1)/2$, so it is $6$ for $d=4$ and $11$ for $d=5$. The discrepancy can be explained noticing that not all directions in the space of probabilities are ``equal'', since the thermal state introduces a thermodynamic asymmetry. Nevertheless, to a first approximation, what the numerics confirm is that Theorem~\ref{thm:finite} should hold by replacing ``enumerate all canonical sequences'' with ``enumerate all canonical sequences of (close to) minimal length''. This suggests an improved worst-case scaling of $O(d^{d^2})$.\footnote{It is however worth emphasising that our algorithm is not using any heuristics: it verifies conditions equivalent to Theorem~\ref{thm:finite}. Here, we are merely discussing how it does so much more quickly that one might naively expect.} Finally, numerics also suggest that the path termination procedure is much more efficient than $O(d^{d^2})$ scaling: the maximum number of paths $N$ surviving at any point in the algorithm (averaged over random initial conditions and thermal distributions) is $N = 9$ for $d=4$ and $N\approx 81.1$ for $d=5$.

These are the underlying reasons why our algorithm allows one to relatively quickly solve problems up to $d=7$, which appears prohibitively large if one looks at simple complexity upper bounds. Despite these strong improvements, pushing this dimension up even further will likely require the application of heuristic approaches to the search problem.

%------------------------------------------------------------------
% SEC. VII.B
%------------------------------------------------------------------

\subsection{Constructive protocol}
\label{sec:constructive}

Following Theorem~\ref{thm:finite}, suppose we find a sequence of elementary thermalizations mapping $\v{p}$ to a final state $\v{f}$ with the same $\v{\gamma}$-ordering as the target $\v{q}$ and satisfying $\v{f} \succ_{\v{\gamma}} \v{q}$. Then, one can explicitly construct a sequence of  of $M \leq d-1$ elementary thermalizations transforming $\v{f}$ into $\v{q}$ (see proof of Theorem~12 in Supplemental Material of Ref.~\cite{perry2018sufficient}). We thus conclude that a sequence of Markovian thermal processes achieving the transformation from $\v{p}(0)$ to $\v{p}(t_f)=\v{q}$ is
\begin{equation}
	\v{q}=\prod_{s=1}^{M}  T^{(f_s,f'_s)}(\lambda_s) 	\prod_{k=1}^{N} T^{(i_k,j_k)}(1) \v{p}(0),
\end{equation}
where the first $N$ elementary thermalizations are obtained from a sequence that satisfies Eq.~\eqref{eq:finitesequence}, and the construction of the remaining ones can be found in Ref.~\cite{perry2018sufficient}. This gives an explicit construction involving a sequence of $N+M \leq d! + d -2$ elementary thermalizations achieving any allowed transformation. In other words, Theorem~\ref{thm:universality} is strengthened to a result that also satisfies desideratum~\ref{d2} of Sec.~\ref{sec:approach}:
\begin{cor}[Strengthened universality of elementary thermalizations]
	\label{cor:strengthened}
	\mbox{$\v{p}(0) \stackrel{ \textrm{MTP}}{\longmapsto} 	\v{p}(t_f)$} if and only if there exists a finite sequence of elementary thermalizations such that 
	\begin{equation}
		\prod_{s=1}^{M}  T^{f_s,f'_s}(\lambda_s) 	\prod_{k=1}^{N} T^{i_k,j_k}(1) \v{p}(0) = \v{p}(t_f),
	\end{equation}
where $N \leq d!-1$, $M \leq d-1$ and the sequence can be algorithmically constructed.
\end{cor}

%------------------------------------------------------------------
% SEC. VII.C
%------------------------------------------------------------------

\subsection{Comparison with previous works}
\label{sec:comparison}

We want to emphasize that the two desiderata considered here,~\ref{d1}-\ref{d2} in Sec.~\ref{sec:approach}, important as they are to systematically develop and optimize explicit protocols in nonequilibrium thermodynamics, are not typically met by general frameworks. The standard framework based on entropy production provides thermodynamic constraints~\cite{landi2021irreversible}, but as discussed in Sec.~\ref{sec:insufficiency} these are insufficient to characterize the future evolution. Furthermore, these entropic relations do not provide tools to construct and optimize nonequilibrium thermodynamic protocols, unless one focuses on special classes of dynamics, or restricted regimes such as close to equilibrium transformations and slow driving protocols~\cite{abiuso2020geometric}. Hence, neither of the two desiderata is satisfied. The same holds true for prominent results in the resource-theoretic approach to thermodynamics. The second laws constraints of Ref.~\cite{brandao2013second} are neither finitely checkable, nor they provide a way to construct explicit protocols. The same applies to the ``quantum majorization constraints'' for thermal processes, the main result derived in Ref.~\cite{gour2018quantum}. 

Several other settings satisfy desideratum~\ref{d1}, but not ~\ref{d2}. The framework based on correlating catalysis~\cite{lostaglio2015stochastic, mueller2018correlating} provides finitely checkable conditions, but not explicit protocols. Similarly, the thermomajorization constraints of Ref.~\cite{horodecki2013fundamental} are finitely checkable by linear programming, but the construction of explicit operations achieving the transformations is not known. More precisely, Ref.~\cite{Alhambra2019heatbathalgorithmic} did provide a set of thermal processes, called $\beta$-permutations, whose convex combination realizes the most general transformation. However, beyond the case $d=2$, it is an open question how these processes can be realized by explicit system-bath interactions. The same is true for the work of Ref.~\cite{shiraishi2021quantum}, which provides Gibbs-preserving channels realizing general transformations in the correlating catalyst setting, but not their explicit thermodynamic realization. All these are examples where the desideratum~\ref{d2} is not satisfied. Despite noticeable progress~\cite{lipka2021all, henao2021catalytic}, the difficulty of satisfying both desiderata has hindered the application of information thermodynamics frameworks.  

Ref.~\cite{perry2018sufficient} is an exception in this regard, as it provides constructive protocols to realize all the transformations allowed by thermal processes. However, it allows a much more extensive control of the system than we considered here, including full control over its energy spectrum and that of a qubit ancilla. These controls are very powerful: as the authors showed, they allow one to achieve every state transformation possible under thermal processes, including those requiring arbitrary non-Markovian dynamics. Here, instead, we focused on the characterization of thermalizations described by Markovian master equations and involving limited control over the energy spectrum. This paves the way to several applications, as discussed in the accompanying paper~\cite{korzekwa2022optimizing}.

%------------------------------------------------------------------
% SEC. VIII
%------------------------------------------------------------------

\section{Explicit algorithmic verification}
\label{sec:algorithm}

Naturally, to verify the complete second laws conditions in higher dimensions, one would like to develop an explicit algorithm that exhaustively verifies the inequalities in Theorem~\ref{thm:finite}. Here we propose one such algorithm with two variants, for which we provide a corresponding \texttt{Mathematica} code~\cite{korzekwa2021continuous}. The fast version only verifies whether $\v{p}(0) \stackrel{ \textrm{MTP}}{\longmapsto} \v{p}(t_f)$ for a fixed initial and final state, while the slow version constructs the set of all final states achievable from a given initial state.

\begin{pabox}[label={algorithmcounter}]{Algorithm verifying $\v{p}\ggcurly_{\v{\gamma}}\v{q}$}
	\begin{enumerate}[leftmargin=0.2cm]

		\item \noindent \textbf{Initialize.}

		\begin{enumerate}[leftmargin=0.2cm]
			\item Create a set of states \texttt{Current} that initially contains only the initial state $\v{p}$.
			\item For each $\v{\gamma}$-ordering, indexed by $k$ from 1 to $d!$, create a set \texttt{Optimal[k]}. Initially all sets are empty except for the ones corresponding to $\v{\gamma}$-orderings of $\v{p}$, which contain only $\v{p}$. 

		\end{enumerate}

		\item \textbf{Generate optimal states.}

		\begin{enumerate}[leftmargin=0.2cm]
			
			\item Update \texttt{Current} to contain all states achievable from old \texttt{Current} via full thermalizations between 2 levels adjacent in the $\v{\gamma}$-ordering. 
			
			\item{} [Only for fast version] Remove those elements of \texttt{Current} that do not thermomajorize $\v{q}$.
			
			\item Denote by \texttt{Current[k]} all states from \texttt{Current} with $\v{\gamma}$-ordering $k$. For each $k$, remove from \texttt{Current} all those states of \texttt{Current[k]} that are thermomajorized by either another state from \texttt{Current[k]} or by any state from \texttt{Optimal[k]}.
			
			\item For each $k$, remove from \texttt{Optimal[k]} all states that are thermomajorized by any state from \texttt{Current[k]}, and then add all states from \texttt{Current[k]} to \texttt{Optimal[k]}.

			\item Repeat steps (a)-(d) until \texttt{Current} is empty.
		\end{enumerate}
		\item \textbf{Verify thermomajorization condition.}
		\begin{enumerate}[leftmargin=0.2cm]
			\item{} [Only for fast version] Verify whether any of the states from \texttt{Optimal[k]}, where $k$ corresponds to $\v{\gamma}$-ordering of $\v{q}$, thermomajorizes $\v{q}$. If yes, then $\v{p}\ggcurly_{{\v{\gamma}}} \v{q}$; otherwise the relation does not hold. 
			\item{} [Only for slow version] The sets \texttt{Optimal[k]} contain all the information about states continuously thermomajorized by $\v{p}$. More precisely,
			 the states with $\v{\gamma}$-ordering $k$ which are continuously thermomajorized by $\v{p}$ are those thermomajorized by \texttt{Optimal[k]}.
		\end{enumerate}
	\end{enumerate}
\end{pabox}

Let us make a few comments on the above algorithm. First, it is clear that it satisfies desideratum~\ref{d2} of finite verifiability, but it can be easily modified to also satisfy desideratum~\ref{d1} of constructability. One simply needs to keep track of the ``history'' of each state: in step~(2a), one should record which elementary thermalization led to a new state. This history should be kept when updating the optimal states with current states in step~(2d). As a result, at the end algorithm we will not only have the list of optimal states within each $\v{\gamma}$-ordering, but we will also know the sequence of elementary thermalizations that need to be applied to the initial state to obtain each of them.

Second, the possible number of canonical sequences grows very fast with the system's dimension $d$. Thus, one may try to develop heuristics to distinguish between better and worse choices of canonical sequences. Using the example from Fig.~\ref{fig:algorithm}, it is intuitively clear that in order to go from $\v{\gamma}$-ordering $\{1,2,3\}$ to $\v{\gamma}$-ordering $\{2,1,3\}$ one should do it directly rather than following the path $\{1,2,3\} \rightarrow \{1,3,2\} \rightarrow \{3,1,2\}\rightarrow  \{3,2,1\}\rightarrow \{2,3,1\}\rightarrow \{2,1,3\}$. Even without any heuristics, we were able to run (on a standard laptop computer) the slow version of the algorithm implemented in \texttt{Mathematica}~\cite{korzekwa2021continuous} to solve the $d=6$ case in minutes and the $d=7$ case in hours. 

%------------------------------------------------------------------
% SEC. IX
%------------------------------------------------------------------

\section{Conclusions and outlook}
\label{sec:conclusions}

In this paper we provided a hybrid framework overcoming current limitations of resource-theoretic and master equation approaches to quantum thermodynamics through the novel notion of continuous thermomajorization. Crucially, our approach includes explicit methods to fully solve the question of the existence of a Markovian thermal process mapping between two non-equilibrium states and returns a corresponding sequence of elementary controls when these exist. Exhaustive searches are feasible on a laptop machine up to $d = 7$ through the \texttt{Mathematica} code we provided~\cite{korzekwa2021continuous}. To achieve this, we employed an exhaustive algorithm (discussed in Sec.~\ref{sec:algorithm}) whose strategy appears to effectively cap the maximum sequence length to (close to) the minimum length, while also curbing the number of active sequences at any given time.

While here we tackled the question of describing an algorithm which is guaranteed to provide a definite answer to the interconversion problem, in many circumstances it is enough to find a method that is able to construct useful working protocols in most cases. A promising direction to probe higher-dimensional systems is then to relax the exhaustive search to a heuristic search. For example, if we are interested in states with a target final $\v{\gamma}$-order $\v{\pi}_f$, one could heuristically restricts to canonical sequences that at each step decrease the transposition distance between the current state's $\v{\gamma}$-order $\v{\pi}_c$ and $\v{\pi}_f$. Another promising direction is to focus on specific task of special relevance, such as cooling. A heuristic protocol is tasked, for example, with decreasing the average energy of the final state. Then, even if an exhaustive search is out of reach, we could run the protocol and terminate it after a given time has passed, or when a given energy target has been met. Coarse-grainings form another direction to be looked at, since we may be interested in the broad properties of the final energy distribution rather than in the exact population in each microscopic energy state. A combination of problem relaxation and heuristic techniques offer, in our opinion, the best way forward to tackling high-dimensional problems.
 
Pushing the achievable dimension up, and combining the current algorithm optimizing the thermalization stage with alternative methods to optimize unitary stages, will likely open up a range of applications, such as the optimization of quantum thermodynamic cycles of heat engines. In the accompanying paper~\cite{korzekwa2022optimizing} we already discuss ways of employing the framework developed here to construct provably optimal thermodynamic protocols. 

At the same time, our framework also offers a rigorous information-theoretical foundation to the dynamical viewpoint of quantum thermodynamics. This approach complements the master equation toolbox, as we have seen, for example, with the systematic construction of generalized entropy production inequalities. Another direction that should be further explored concerns the role of quantum coherence in these settings. We provide some initial remarks in Appendix~\ref{app:coherence}, while a solution to this problem satisfying both desiderata~\ref{d1}-\ref{d2} is still out of reach.    

%------------------------------------------------------------------
% ACKNOWLEDGEMENTS
%------------------------------------------------------------------

\section*{Acknowledgements}

M.L. thanks A. Levy for useful discussions. K.K. acknowledges financial support by the Foundation for Polish Science through TEAM-NET project (contract no. POIR.04.04.00-00-17C1/18-00). M.L. acknowledges financial support from the the European Union's Marie Sk{\l}odowska-Curie individual Fellowships (H2020-MSCA-IF-2017, GA794842), Spanish MINECO (Severo Ochoa SEV-2015-0522 and project QIBEQI FIS2016-80773-P), Fundacio Cellex and Generalitat de Catalunya (CERCA Programme and SGR 875) and grant EQEC No. 682726.

\appendix

%------------------------------------------------------------------
% APPENDIX A
%------------------------------------------------------------------

\section{Proof of Theorem~\ref{thm:main} and Theorem~\ref{thm:universality}}
\label{app:generaltheorems}

To build up towards our final results we will need several intermediate technical statements. We start by recalling an important result derived in Ref.~\cite{perry2018sufficient}.

\begin{thm}[Theorem~12, Supplemental Material of Ref.~\cite{perry2018sufficient}]
	\label{thm:perry}
	If $\v{p} \succ_{\v{\gamma}}\v{q}$ and $\v{\pi}(\v{p}) = \v{\pi}(\v{q})$, there exists a sequence of elementary thermalizations $\{T^{i_k,j_k}(\lambda_k)\}_{k=1}^f$ such that 
	\begin{equation}
	T^{i_f,j_f}(\lambda_f) \dots T^{i_1,j_1}(\lambda_1) \v{p} = \v{q}.
	\end{equation}
	Moreover, $f \leq d-1$.
\end{thm}

We will also need the following known result characterizing thermomajorization
	
	\begin{thm}[See Refs.~\cite{horodecki2013fundamental, ruch1978mixing, lostaglio2019introductory}]
		\label{thm:thermomajorizationprevious} 
		There exists a stochastic matrix $T$ such that $T \v{p} = \v{q}$ and $T \v{\gamma} = \v{\gamma}$ if and only if  $\v{p} \succ_{\v{\gamma}} \v{q}$.
	\end{thm}

Next, we link continuous thermomajorization between two distributions with the existence of a sequence of elementary thermalizations bringing one distribution to another.

\begin{lem}[Continuous thermomajorization and elementary thermalizations]
	\label{lem:plt}
	$\v{p} \ggcurly_{\v{\gamma}} \v{q}$ if and only if there exists a finite sequence of elementary thermalizations $\{T^{i_k,j_k}(\lambda_k)\}_{k=1}^f$ such that 
	\begin{equation}
	\label{eq:sequenceplt}
	T^{i_f,j_f}(\lambda_f) \dots T^{i_1,j_1}(\lambda_1) \v{p} = \v{q}.
	\end{equation}
\end{lem}
\begin{proof}
	First, assume $\v{p} \ggcurly_{\v{\gamma}} \v{q}$. Then, there exists a continuous trajectory $\v{r}(t)$ with $\v{r}(0) = \v{p}$, $\v{r}(t_f) = \v{q}$ (perhaps $t_f = +\infty$), and $\v{r}(t') \succ_{\v{\gamma}} \v{r}(t'')$ for all $t' \leq t''$. Define $t_0=0$, as well as a thermomajorization ordering $\v{\pi}^1$ and a time $t_1$ as follows: 
	\begin{subequations}
		\begin{align}
			\v{\pi}^1 :=& \v{\pi}(\v{r}(0)),\\
			t_1 :=& \sup\{t|\v{\pi}(\v{r}(t)) = \v{\pi}^1  \}.	
		\end{align}
	\end{subequations}
	Next, for integer $k> 1$, define iteratively
	\begin{subequations}
		\begin{align}
		\v{\pi}^{k+1} := &\v{\pi}(\v{r}(t^+_k)),\label{eq:crossing_a}\\
		t_{k+1} := &\sup\{t|\v{\pi}(\v{r}(t)) = \v{\pi}^{k+1}\}.\label{eq:crossing_b}
		\end{align}
	\end{subequations}
	Clearly $t_{k+1} > t_{k}$. Since there are only $d!$ distinct thermomajorization orderings, ultimately we reach the final $k=f\leq d!-1$, such that $\v{\pi}^{f} = \v{\pi}(\v{r}(t_f))$. We now employ Theorem~\ref{thm:perry}: for each pair, $\v{r}(t_k^+)$ and $\v{r}(t_{k+1})$, there exists a sequence of elementary thermalizations such that
	\begin{equation}
	\v{r}(t_{k+1}) = 	T^{i_{k_n},j_{k_n}}(\lambda_{k_n}) \dots T^{i_{k_1},j_{k_1}}(\lambda_{k_1}) \v{r}(t_k),
	\end{equation}
	with $n \leq d-1$. Thus, by sequentially applying the above to all $k\leq f$ we obtain Eq.~\eqref{eq:sequenceplt} with a finite sequence.
	
	Conversely, assume that Eq.~\eqref{eq:sequenceplt} holds. Define \mbox{$\v{r}(0) = \v{p}$} and
	\begin{equation*}
	\v{r}(t) = T^{i_k,j_k}(\delta) T^{i_{k-1},j_{k-1}}(\lambda_{k-1}) \dots T^{i_{1},j_{1}}(\lambda_1)\v{p}, 
	\end{equation*}      
	with $k$ and $\delta \in [0,1]$ satisfying \mbox{$t = \delta + \sum_{i=1}^{k-1} \lambda_i$}. This defines a continuous path starting at $\v{p}$ and terminating at~$\v{q}$.	Moreover, using the fact that for any $i,j$ and $\lambda'\geq\lambda$ we can write \mbox{$T^{i,j}(\lambda') =T^{i,j}(\mu) T^{i,j}(\lambda)$} with $\mu \in[0,1]$, we see that for any $t''\geq t'$ the distribution $\v{r}(t'')$ is obtained from $\v{r}(t')$ by a finite sequence of elementary thermalizations. Given that elementary thermalizations and their products are stochastic matrices with a fixed point $\v{\gamma}$, Theorem~\ref{thm:thermomajorizationprevious} implies $\v{r}(t') \succ_{\v{\gamma}} \v{r}(t'')$ for all $t''\geq t'$. We thus conclude that $\v{p} \ggcurly_{\v{\gamma}} \v{q}$.
\end{proof}

Note that from the proof above one can conclude that the number  of elementary thermalizations required for a state transformation is upper-bounded by \mbox{$d!(d-1)$}, but we will give a tighter bound later. Also, as a corollary of Lemma~\ref{lem:plt} we get that $\succ_{\v{\gamma}}$ (describing allowed transformations under general thermal processes, which employ with memory) and $\ggcurly_{\v{\gamma}}$ (describing allowed transformations under Markovian thermal processes) coincide within a fixed thermomajorization ordering.

\begin{cor}
	\label{corol:orderingscoincide}
	If $\v{\pi}(\v{p}) = \v{\pi}(\v{q}) $ and $\v{p} \succ_{\v{\gamma}} \v{q}$ then $\v{p} \ggcurly_{\v{\gamma}} \v{q}$. Moreover, the thermomajorizing trajectory $\v{r}(t)$ connecting $\v{p}$ to $\v{q}$ can be chosen such that for all $t\in[0,t_f]$ it belongs to the same $\v{\gamma}$-ordering.
\end{cor}
\begin{proof}
	Assuming that \mbox{$\v{p} \succ_{\v{\gamma}} \v{q}$} and $\v{\pi}(\v{p}) = \v{\pi}{(\v{q})}$, Theorem~\ref{thm:perry} tells us then that there exists a sequence of elementary thermalizations mapping $\v{p}$ into $\v{q}$. Using Lemma~\ref{lem:plt}, we conclude $\v{p} \ggcurly_{\v{\gamma}} \v{q}$. Moreover, the construction of elementary thermalizations presented in the proof of Theorem~\ref{thm:perry} in Ref.~\cite{perry2018sufficient} is such that every intermediate state along the trajectory, $\v{r}(t)$, has the same thermomajorization ordering $\v{\pi}(\v{p})$. 
\end{proof}

	We are now able to discuss Theorem~\ref{thm:main} and Theorem~\ref{thm:universality}, which we prove jointly as follows:
	\begin{thm}
		Let $\rho(t)$ by a quantum state with population vector $\v{p}(t)$. The following statements are equivalent:
		\begin{enumerate}
			\item There exists a Markovian thermal process transforming $\rho(0)$ with population $\v{p}(0)$ into a quantum state $\rho(t_f)$ with population $\v{p}(t_f)$. 
			\item $\v{p}(0) \ggcurly_{\v{\gamma}} \v{p}(t_f)$.
			\item There exists a finite sequence of elementary thermalizations such that 
			\begin{equation}
			\label{eq:sequenceelementarythermalizationsproof}
			T^{i_f,j_f}(\lambda_f) \dots T^{i_1,j_1}(\lambda_1) \v{p}(0) = \v{p}(t_f).
			\end{equation}
		\end{enumerate}
	\end{thm}
	\begin{proof}
		Lemma~\ref{lem:plt} proves the equivalence $2 \Leftrightarrow 3$. To conclude we then just prove $3 \Rightarrow 1$ and $1 \Rightarrow 2$.
		
		\bigskip
		
		[$3 \Rightarrow 1$]: Given the finite sequence of elementary thermalizations such that Eq.~\eqref{eq:sequenceelementarythermalizationsproof} holds, we will explicitly construct a time-dependent Lindbladian $\mathcal{L}_t$ generating a Markovian thermal process that maps a state with population $\v{p}(0)$ to the one with population $\v{p}(t_f)$. We define $\mathcal{L}_t$ through its action on the basis elements as
		\begin{equation}
			\bra{m}\mathcal{L}_t(\ketbra{n}{n'})\ket{m'}  = \delta_{nn'} \delta_{mm'} T(t)_{mn} - \delta_{mn} \delta_{m' n'},
		\end{equation}
		where $\ket{x}$ denote the eigenstates of $H$ and $T(t)$ is a $d\times d$ stochastic matrix. To see that the above $\mathcal{L}_t$ corresponds to a valid Lindbladian, note that it has the form \mbox{$\mathcal{E}_t - \mathcal{I}$}, where $\mathcal{I}$ is the identity channel and $\mathcal{E}_t$ is the channel that decoheres in the eigenbasis of $H$ and performs the stochastic map $T(t)$ on the diagonal. Since every channel can be written as $\mathcal{E}_t(\rho) = \sum_i L_i \rho L_i^\dagger$ with $\sum_i L_i^\dagger L_i = \mathbb{I}$, $\mathcal{L}_t$ has the required form from Eq.~\eqref{eq:lindbladian}.
			
		Next, let
		\begin{equation}
			t_k = - \log(1-\lambda_k),\quad t_0 := 0
		\end{equation}
		and introduce $\tau_k = \sum_{s=0}^k t_k$. Then, choose
		\begin{equation}
			T(t) = 
			\begin{bmatrix}
				\frac{\gamma_{i_k}}{\gamma_{i_k} + \gamma_{j_k}}  & \frac{\gamma_{i_k}}{\gamma_{i_k} + \gamma_{j_k}} \\ \\
				\frac{\gamma_{j_k}}{\gamma_{i_k} + \gamma_{j_k}} & \frac{\gamma_{j_k}}{\gamma_{i_k} + \gamma_{j_k}}
			\end{bmatrix} \oplus \v{0}_{\backslash (i_k, j_k)}
		\end{equation}
		for $t \in \left[\tau_{k-1}, \tau_{k}\right]$ with $k=1, \dots, f$, where $\v{0}_{\backslash (i_k, j_k)}$ denotes the $(d -2) \times (d-2)$ matrix of all zeros acting on the subspace of all energy levels except $i_k, j_k$. Now, the equation \mbox{$d\rho/dt = \mathcal{L}_t(\rho_t)$} can be easily solved, and one can verify that the resulting dynamics implements the sequence of elementary thermalizations from Eq.~\eqref{eq:sequenceelementarythermalizationsproof} on the population vector. We conclude that a Markovian thermal process mapping  $\v{p}(0)$ into $\v{p}(t_f)$ exists. 
		
		\bigskip
		
		$[1 \Rightarrow 2]:$ Given $\v{p}(0)$, $\v{p}(t_f)$ and a Markovian thermal process dynamically evolving the former into the latter, let $\v{p}(t)$ be the trajectory followed at the intermediate times. Employing the definition of an MTP, we have that for every \mbox{$0 \leq t' \leq t'' \leq t_f$} there exists a stochastic matrix $T(t',t'')$ such that 
		\begin{equation}
		T(t',t'') \v{p}(t') = \v{p}(t''), \quad T(t',t'') \v{\gamma} = \v{\gamma}.
		\end{equation}
		From Theorem~\ref{thm:thermomajorizationprevious}, it follows that $\v{p}(t') \succ_{\v{\gamma}} \v{p}(t'')$. We conclude that $\v{p} \ggcurly_{\v{\gamma}} \v{q}$. 
	\end{proof}

%------------------------------------------------------------------
% APPENDIX B
%------------------------------------------------------------------

\section{Proof of Theorem~\ref{thm:finite}}
\label{app:thm4}

We start from introducing the concept of a \emph{coarse-grained description} of a thermomajorizing trajectory $\v{r}(t)$.

\begin{defn}[Coarse-graining]
	\label{def:coarsegraining}
	Let $\v{r}(t)$ be a thermomajorizing trajectory from $\v{p}$ to $\v{q}$ as in Definition~\ref{def:markov_majo}. Then, a sequence of $\v{\gamma}$-orderings $\{\v{\pi}^k\}_{k=1}^N$ is called a coarse-grained description of $\v{r}(t)$ if there exists an ordered set of times $\{t_k\}_{k=0}^N$,
	\begin{equation}
	t_0=0,\quad t_k\leq t_{k+1},\quad t_N=t_f,
	\end{equation}	
	such that the probability vector $\v{r}(t)$ belongs to the $\v{\gamma}$-ordering $\v{\pi}^k$ in the interval $t\in[t_{k-1},t_k]$ (note that at time $t_k$, $\v{r}(t_k)$ is associated to both $\v{\gamma}$-orders $\v{\pi}^k$ and $\v{\pi}^{k+1}$). Moreover, a given coarse-grained description will be called \emph{canonical} if the sequence $\{\v{\pi}^k\}$ is canonical according to Definition~\ref{def:canon_seq}. 
\end{defn}

Note that since a given state can simultaneously belong to more than one $\v{\gamma}$-ordering (this happens when $p_k/\gamma_k=p_l/\gamma_l$ for some $k$ and $l$), the trajectory $\v{r}(t)$ can have multiple (but finitely many) coarse-grained descriptions. However, as we now prove, we may limit our attention only to canonical coarse-grainings.

\begin{lem}
	\label{lem:canonical_path}
	$\v{p} \ggcurly_{\v{\gamma}} \v{q}$ if and only if there exists a thermomajorizing path $\v{r}(t)$ connecting $\v{p}$ and $\v{q}$ with a canonical coarse-grained description.
\end{lem}

\begin{proof}
	Assume that $\v{p} \ggcurly_{\v{\gamma}} \v{q}$, meaning that there exists some thermomajorizing trajectory $\v{r}(t)$ from $\v{p}$ to $\v{q}$. Since $\v{r}(t)$ changes continuously in $t$, so do $\v{\gamma}$-rescaled entries $\{r_i(t)/\gamma_i\}_{i=1}^d$. This means that the $k$-th largest element among $\v{\gamma}$-rescaled entries $\{r_i(t)/\gamma_i\}_{i=1}^d$ becomes equal to the $(k-1)$-th or $(k+1)$-th largest one before (or at the same time) becoming equal to any other entry.  Hence, we can choose times $\{t_k\}_{k=0}^N$ (perhaps $t_{k} = t_{k+1}$ for some $k$) and assign $\v{\gamma}$-orderings $\{\v{\pi}^k\}_{k=1}^N$ such that $\v{\pi}^k$ and $\v{\pi}^{k+1}$ only differ by a transposition of adjacent elements, which proves the first property of the canonical coarse-graining. 
	
	To prove the second one, assume that some $\v{\gamma}$-ordering $\v{\pi}$ appears more than once in the coarse-grained description of $\v{r}(t)$ defined in the previous step. Let us denote the first time $\v{r}(t)$ has the ordering $\v{\pi}$ by $t_1$, and the last time it has this ordering by $t_2$. Clearly, $\v{r}(t_1) \ggcurly_{\v{\gamma}} \v{r}(t_2)$. Thus, one can introduce a new thermomajorizing path $\v{r}'(t)$ such that it is equal to $\v{r}(t)$ for $t\in[0,t_1)$ and $t\in(t_2,t_f]$, while for $t\in[t_1,t_2]$ it is the trajectory given by Corollary~\ref{corol:orderingscoincide}, connecting $\v{r}(t_1)$ to $\v{r}(t_2)$ and lying completely in $\v{\gamma}$-ordering $\v{\pi}$. As a result, we obtain a thermomajorizing trajectory $\v{r}'(t)$ connecting $\v{p}$ and $\v{q}$, and such that the $\v{\gamma}$-ordering $\v{\pi}$ appears only once in the coarse-grained description. By repeating this for every $\v{\gamma}$-ordering that appears more than once in the original coarse-grained description of $\v{r}(t)$, we end up with a trajectory whose coarse-grained description satisfies also the second property of canonical coarse-graining.
	
	Conversely, if there exists a thermomajorizing path $\v{r}(t)$ connecting $\v{p}$ and $\v{q}$ with whatever coarse-grained description, then by definition $\v{p} \ggcurly_{\v{\gamma}} \v{q}$. 
\end{proof}

The next lemma geometrically characterizes the action of elementary thermalizations on a Lorenz curve/thermomajorization curve. In words, it shows that the effect of $T^{i,j}$ is to decrease the slope of the $j^{\rm th}$ segment of the thermomajorization curve and increase that of the $i^{\rm th}$ segment till the two are equalized (see Fig.~\ref{fig:lemma10} and Appendix~B.1 of Ref.~\cite{perry2018sufficient}).
\begin{figure}[t]
	\centering
	\includegraphics[width=0.9\columnwidth]{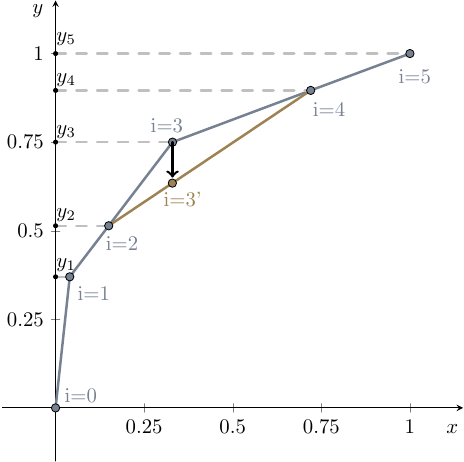}
	\caption{The defining points of the thermomajorization curve are labelled according to Lemma~\ref{lem:action}, to visualize the action of the two-level partial level thermalization $T^{3,4}(\lambda)$. This brings down $y_3$ until, at $\lambda=1$, the slopes of the $3^{rd}$ and $4^{th}$ segment are equalized (brown segment connecting $2$, $3'$ and $4$).}
	\label{fig:lemma10}
\end{figure}

\begin{lem}[Action of elementary thermalization on the thermomajorization curve]
	\label{lem:action}
	Consider a Lorenz curve of $\v{p}$ with elbow points denoted by $(x_m, y_m)$, see Fig.~\ref{fig:lemma10}. For any $j>i$, denote $\tilde{i}=\pi_i(\v{p})$ and $\tilde{j}=\pi_j(\v{p})$. Then, the elementary thermalization $T^{\tilde{i},\tilde{j}}(\lambda)$ shifts down by an equal amount the $y$-coordinates $(y_{i},\dots,y_{j-1})$. The extremal map, \mbox{$T^{\tilde{i},\tilde{j}}(1)$}, equalizes the slopes of the $i^{th}$ and the $j^{th}$ segments of the curve. Note that a final reordering may be needed if the thermomajorization ordering is changed.
\end{lem} 
\begin{proof}
	For $\v{p}' = T^{\tilde{i},\tilde{j}}(\lambda)\v{p}$ we have
	\begin{equation}
	\!\!\!	p'_m = \left\{\begin{array}{ll}
	p_m \quad& \textrm{for } m \notin\{\tilde{i},\tilde{j}\}, \\
	(1-\lambda) p_m +\lambda\frac{p_{\tilde{i}}+p_{\tilde{j}}}{\gamma_{\tilde{i}}+\gamma_{\tilde{j}}}\gamma_m  \quad&\textrm{for } m \in \{\tilde{i},\tilde{j}\}.
	\end{array}\right.
	\end{equation}
	Denote by $y_m$ and $y'_m$ the $y$-coordinates of the thermomajorization curves of $\v{p}$ and $\v{p}'$, respectively. Then
	\begin{equation}
	y'_m = \left\{\begin{array}{ll}
	y_m \quad &\textrm{for } m < i, \\
	y_m - \lambda \frac{( p_{\tilde{i}} \gamma_{\tilde{j}} -p_{\tilde{j}} \gamma_{\tilde{i}})}{\gamma_{\tilde{i}}+\gamma_{\tilde{j}}} \quad &\textrm{for } i\leq m  < j, \\
	y_m \quad& \textrm{for } m \geq j.
	\end{array}\right.
	\end{equation}
	This corresponds to shifting down the $y$-coordinate of each point of the thermomajorization curve, starting from the $i^{th}$ point to the $(j-1)^{th}$ point. Setting $\lambda =1$ and using \mbox{$y_j-y_{j-1}=p_{\tilde{j}}$}, one obtains that the slope of the $j^{th}$ segment is
	\begin{equation*}
	\frac{y'_j - y'_{j-1}}{\gamma_{\tilde{j}}} = \frac{p_{\tilde{i}} + p_{\tilde{j}}}{\gamma_{\tilde{i}}+ \gamma_{\tilde{j}}}. 
	\end{equation*}
	Similarly, the slope of the $i^{th}$ segment is
	\begin{equation*}
	\frac{y'_i - y'_{i-1}}{\gamma_{\tilde{i}}} = \frac{p_{\tilde{i}} + p_{\tilde{j}}}{\gamma_{\tilde{i}} + \gamma_{\tilde{j}}}. 
	\end{equation*}
	Hence, the two slopes are equalized for $\lambda =1$. Note that, if (and only if) $j \neq i+1$, then the thermomajorization ordering will change at some intermediate $\lambda$, so that a rearrangement of the segments is necessary to sort them according to non-increasing slopes.
\end{proof}

The next lemma identifies states obtained by thermalizing two adjacent levels in the $\v{\gamma}$-ordering as the optimal crossings between one ordering and another.
\begin{lem}
	\label{lem:bestcrossing}
	Given $\v{p}$, we have 
	\begin{equation}
		T^{\pi_i(\v{p}),\pi_{i+1}(\v{p})}(1)\v{p} \ggcurly_{\v{\gamma}} \v{q}
	\end{equation}
	for every $\v{q}$ satisfying $\v{p} \ggcurly_{\v{\gamma}} \v{q}$, $\v{\pi}(\v{q})= \v{\pi}(\v{p})$ and
	\begin{equation}
	\label{eq:equalizeslopes}
	\frac{q_{\pi_i(\v{p})}}{\gamma_{\pi_i(\v{p})}} = 	\frac{q_{\pi_{i+1}(\v{p})}}{\gamma_{\pi_{i+1}(\v{p})}}.
	\end{equation} 
	
\end{lem}
\begin{proof}
	Since $\v{p} \ggcurly_{\v{\gamma}} \v{q}$ implies $\v{p} \succ_{\v{\gamma}} \v{q}$, and furthermore we have $\v{\pi}(\v{q})= \v{\pi}(\v{p})$, we can use Theorem~\ref{thm:perry}. Its proof shows that there is a sequence of elementary thermalizations such that 
	\begin{equation}
	\v{q} = T^{i_k,j_k}(\lambda_k) \cdots T^{i_1,j_1}(\lambda_1)\v{p},
	\end{equation}
	and, furthermore, that all intermediate states have thermomajorization ordering equal to $\v{\pi}(\v{p})$. Then, we conclude from Lemma~\ref{lem:action} that the thermomajorization curve of $\v{q}$ can be obtained from that of $\v{p}$ by lowering a set of its $y$-coordinates while making the slopes of the $i$-th and $(i+1)$-th segments equal (no re-orderings are involved, so the $x$-coordinates do not change).
	
	Now let us focus on $\v{q}^{i,i+1}:=	T^{\pi_i(\v{p}),\pi_{i+1}(\v{p})}(1)\v{p}$. Using again Lemma~\ref{lem:action}, the thermomajorization curve of the state $\v{q}^{i,i+1}$ is obtained by lowering the coordinate $y_{i+1}$ till the slopes of the $i$-th and $(i+1)$-th segments are equal, while leaving all other $y$-coordinates untouched.  Crucially, note that this is the minimal $y$-coordinate lowering for any $\v{q}$ that ensures Eq.~\eqref{eq:equalizeslopes} is satisfied. From this property and the fact that $\v{\pi}(\v{q})= \v{\pi}(\v{q}^{i,i+1})$, it follows that $\v{q^{i,i+1}} \succ_{\v{\gamma}} \v{q}$. That is because, since $\v{\pi}(\v{q})= \v{\pi}(\v{q}^{i,i+1})$ implies that the $x$-coordinates of the thermomajorization curves of $\v{q}$ and $\v{q}^{i,i+1}$ coincide, $\succ_{\v{\gamma}}$ is equivalent to simply comparing the $y$-coordinates of the respective thermomajorization curves.  Using Corollary~\ref{corol:orderingscoincide} we conclude $\v{q}^{i,i+1} \ggcurly_{\v{\gamma}} \v{q}$. 
\end{proof}

We are now ready to prove the central lemma from which the proof of Theorem~\ref{thm:finite} follows almost directly.

\begin{lem}
	\label{lem:optimal_path}
	Given a canonical sequence $\{\v{\pi}^k\}_{k=1}^{N}$ with $\v{\pi}^1=\v{\pi}(\v{p})$ and $\v{\pi}^N=\v{\pi}(\v{q})$, the following two statements are equivalent:
	\begin{enumerate}
		\item There exists a thermomajorizing trajectory $\v{r}(t)$ from $\v{p}$ to $\v{q}$ with a canonical coarse-grained description $\{\v{\pi}^k\}_{k=1}^{N}$.
		\item The following relation holds:
		\begin{equation}
		\label{eq:plttrajectory2}
		\prod_{k=1}^{N-1} T^{i_k,j_k}(1) \v{p} \succ_{\v{\gamma}} \v{q},
		\end{equation}
		where $T^{i_k,j_k}(1)$ are full thermalizations between levels $i_k$ and $j_k$ specified by the adjacent elements that differ between $\v{\pi}^k$ and $\v{\pi}^{k+1}$.
	\end{enumerate}
\end{lem}
\begin{proof}
	We will prove that the implication holds both ways.
	
	\bigskip
	
	$[1\Rightarrow 2]$ Let $t_k$ for $k\in\{0,\dots,N\}$ be the times from Definition~\ref{def:coarsegraining}, i.e., $t_0=0$, $t_N=t_f$ and the remaining ones describing times for which $\v{r}(t)$ changes $\v{\gamma}$-ordering. By definition, there exist indices $s_k$ for $k\in\{1,\dots,N-1\}$ such that $\v{\pi}^k$ and $\v{\pi}^{k+1}$ differ only by a transposition of the two adjacent entries, $s_k$ and $s_{k}+1$. Let us denote the corresponding energy levels that change order in the $\v{\gamma}$-ordering by:
	\begin{equation}
		i_{k}:= \pi_{s_k}(\v{r}(t_{k})), \quad j_{k}:= \pi_{s_k+1}(\v{r}(t_{k})).
	\end{equation}  
	
	By assumption, for $k\in\{0,\dots,N-1\}$ we have $\v{r}(t_{k})  \ggcurly_{\v{\gamma}} \v{r}(t_{k+1})$ and
	\begin{equation}
		\v{\pi}(\v{r}(t_{k}))= \v{\pi}^{k+1}, \quad \v{\pi}(\v{r}(t_{k+1}))=\v{\pi}^{k+1}.	
	\end{equation}
	Furthermore, for $k\in\{0,\dots,N-2\}$
	\begin{equation}
		\frac{r(t_{k+1})_{i_{k+1}}}{\gamma_{i_{k+1}}} = 	\frac{r(t_{k+1})_{j_{k+1}}}{\gamma_{j_{k+1}}}.
	\end{equation} 
	We are then in the conditions to apply Lemma~\ref{lem:bestcrossing} to conclude that for $k\in\{0,\dots,N-2\}$ we have
	\begin{equation}
		T^{i_{k+1},j_{k+1}}(1)\v{r}(t_k) \ggcurly_{\v{\gamma}} \v{r}(t_{k+1}).
	\end{equation}
	Applying the above sequentially we arrive at
	\begin{align}
		\v{r}(t_{N-1}) & \llcurly_{\v{\gamma}}   T^{i_{N-1},j_{N-1}}(1)\v{r}(t_{N-2}) \nonumber\\
		& \llcurly_{\v{\gamma}}  T^{i_{N-1},j_{N-1}}(1) T^{i_{N-2},j_{N-2}}(1)\v{r}(t_{N-3}) \nonumber\\
		& \llcurly_{\v{\gamma}} \dots \nonumber\\
		& \llcurly_{\v{\gamma}} \prod_{k=1}^{N-1} T^{i_k,j_k}(1) \v{r}(t_0=0).
	\end{align}
	Now, recall that $\v{r}(t_{N-1})	\ggcurly_{\v{\gamma}} \v{r}(t_{N}) = \v{q}$ and that \mbox{$\v{r}(0) = \v{p}$}. Finally, using transitivity and the fact that relation $\ggcurly_{\v{\gamma}}$ implies the relation $\succ_{\v{\gamma}}$, we conclude that
	\begin{equation}
		\prod_{k=1}^{N-1} T^{i_k,j_k}(1) \v{p} \succ_{\v{\gamma}} \v{q}.
	\end{equation}

		\bigskip
		
	$[2\Rightarrow 1]$ For $t \in [0, N-1]$ define the trajectory 
	\begin{equation}
		\v{r}(t) = T^{i_k,j_k}(\delta) T^{i_{k-1},j_{k-1}}(\lambda_{k-1}) \dots T^{i_{1},j_{1}}(\lambda_1)\v{p}, 
	\end{equation}      
	with $k$ and $\delta \in [0,1]$ satisfying \mbox{$t = \delta + \sum_{i=1}^{k-1} \lambda_i$}. Note that the resulting trajectory $\v{r}(t)$ has a coarse-grained description given by $\{\v{\pi}^k\}_{k=1}^{N}$ and that
	\begin{equation}
		\v{p} \ggcurly_{{\v{\gamma}}} \v{r}(N-1).
	\end{equation} 
	Now, from the assumption we have that \mbox{$\v{r}(N-1) \succ_{\v{\gamma}} \v{q}$} and \mbox{$\v{\pi}(\v{r}(N-1)) = \v{\pi}(\v{q})$}. By Corollary~\ref{corol:orderingscoincide} we thus have that $\v{r}(N-1)	\ggcurly_{\v{\gamma}}  \v{q}$, and that the trajectory connecting these two states has a fixed $\v{\gamma}$-ordering. Hence, $\v{p} \ggcurly_{{\v{\gamma}}} \v{r}(N-1) \ggcurly_{{\v{\gamma}}} \v{q}$ and, by transitivity, $\v{p}  \ggcurly_{{\v{\gamma}}}\v{q}$. Therefore, we conclude that there exists a thermomajorizing trajectory from $\v{p}$ to $\v{q}$ with a canonical coarse-grained description given by $\{\v{\pi}^k\}_{k=1}^{N}$.
\end{proof}

We are now ready to present the proof of Theorem~\ref{thm:finite}.

\begin{proof}
	First, assume that $\v{p} \ggcurly_{\v{\gamma}} \v{q}$. Then, from Lemma~\ref{lem:canonical_path}, we know that there exists a thermomajorizing path $\v{r}(t)$ connecting $\v{p}$ and $\v{q}$ with some canonical coarse-grained description. Moreover, from Lemma~\ref{lem:optimal_path} we know that such a path with a given coarse-grained canonical description exists if and only if Eq.~\eqref{eq:plttrajectory2} is satisfied. Thus, if we list all possible canonical coarse-grained paths connecting $\v{p}$ and $\v{q}$, then for at least one of them Eq.~\eqref{eq:plttrajectory2} must be satisfied.
	
	Conversely, if Eq.~\eqref{eq:plttrajectory2} is satisfied for some canonical coarse-grained path, then from Lemma~\ref{lem:optimal_path} we have \mbox{$\v{p} \ggcurly_{\v{\gamma}} \v{q}$}.
\end{proof}
	
%------------------------------------------------------------------
% APPENDIX C
%------------------------------------------------------------------

\section{Remarks on fundamental constraints on coherence}
\label{app:coherence}

In this paper we focused on necessary and sufficient conditions and explicit protocols to generate a given population dynamics. What about the characterization of the evolution of coherences, i.e., the off-diagonal elements of the density matrix in the energy basis? Here we provide some general remarks on this notoriously complex issue. 

It has been recognized that the thermodynamic evolution of quantum coherence is restricted by symmetry considerations~\cite{lostaglio2015description}. These involve an extension of Noether's theorem to open quantum system dynamics~\cite{marvian2014extending}. The properties~\ref{p1}-\ref{p2} then allow to construct monotonically non-increasing or non-decreasing functionals (\emph{monotones}) which quantify the deterioration of athermality and coherence during a thermalization process through entropy production-like inequalities. The standard entropy production can be seen as one monotone, to be accompanied by many more which often have independent operational or information-theoretical meaning. Let us briefly discuss some notable examples.

First, we have the $\alpha$-R\'{e}nyi entropy production relations:
\begin{equation}
	\forall \alpha\in[0,\infty):\quad \frac{d\Sigma_\alpha(t)}{dt} =  -\frac{dS_{\alpha}(\rho(t)\| \gamma)}{dt} \geq 0,
\end{equation}
where $\gamma$ is the thermal state from Eq.~\eqref{eq:thermalstate} and $S_\alpha$ is the quantum $\alpha$-R\'{e}nyi divergence defined by
\begin{equation}
\!\!	S_\alpha(\rho\|\sigma) = \def\arraystretch{3}
	\left\{
	\begin{array}{ll}
		\dfrac{\log \tr{\rho^\alpha \sigma^{1-\alpha}}}{\alpha-1}  & \text{~if} \; \alpha \in [0,1),\\
		\dfrac{\log \tr{ \sigma^{\frac{1-\alpha}{2\alpha}}\rho \sigma^{\frac{1-\alpha}{2\alpha}}}}{\alpha-1}  & \text{~if} \; \alpha >1,
	\end{array}    
	\right.   
\end{equation}
whose origin in quantum information lies in quantum hypothesis testing~\cite{mosonyi2015quantum}. The usual entropy production relation of Eq.~\eqref{eq:entropy_prod} is recovered in the limit $\alpha \rightarrow 1$. Next, we have the $\alpha$-relative entropy of asymmetry $A_\alpha(t)$~\cite{marvianthesis, lostaglio2015description},
\begin{equation}
\!\!\forall \alpha\in[0,\infty):\quad\!\!\!-\frac{d A_\alpha(t)}{dt}:= -\frac{d S_\alpha(\rho(t)\| \mathcal{D}(\rho))}{dt} \geq 0, 
	\label{eq:asymmetry}
\end{equation}
with $\mathcal{D}$ denoting the dephasing in the energy basis. Asymmetry is a measure of coherence of the state in the basis of $H$. For $\alpha =1$ the usual entropy production relation can be decomposed as
\begin{equation}
	\frac{d\Sigma(t)}{dt} = \frac{d\Sigma_d(t)}{dt} - \frac{d A_1(t)}{dt},
\end{equation}
where $\Sigma_d=\sum_i p_i\log(p_i/\gamma_i)$ is the diagonal relative entropy. Hence, $- d A_1(t)/dt\geq 0$ can be seen as one of two monotonically non-decreasing additive terms in the standard entropy production equation, measuring the entropic contribution due to loss of quantum coherence~\cite{lostaglio2015description}. This constitutes a refinement of the usual entropy production inequality, since both $\Sigma_d(t)$ and $-A_1(t)$ do not decrease in time, not just their sum~\cite{santos2019role}. The quantity $A_1(t)$ operationally quantifies the coherent contribution in an average work extraction protocol~\cite{baumer2018fluctuating}, that is the loss to the maximal extractable work due to dephasing of the state. Another example is given by the Fisher information $\mathcal{Q}(\rho(t))$ of the unitary orbit $\{ e^{-i H \delta} \rho(t) e^{i H \delta}\}_{\delta \in \mathbb{R}}$~\cite{janzing2003quasi, kwon2018clock},
\begin{equation}
	-\frac{d\mathcal{Q}(t)}{dt} \geq 0,
\end{equation} 
where 
\begin{equation}
	\mathcal{Q}(t) = 2-2\lim_{\delta \rightarrow 0} F^2[\rho(t), e^{-i H \delta} \rho(t) e^{i H \delta})]
\end{equation}
and $F(\rho,\sigma) := \mathrm{Tr}(\sqrt{\rho^{1/2} \sigma \rho^{1/2}})$ is the quantum fidelity. The quantity $\mathcal{Q}(t)$ measures the information that $\rho(t)$ encodes about the phase along the unitary orbit generated by $H$, i.e. the metrological value of $\rho(t)$ for phase estimation. Finally, we have the Wigner-Yanase-Dyson skew information~\cite{wigner1963information, marvianthesis},
\begin{equation}
	-\frac{d I_s(t)}{dt} \geq 0, 
\end{equation}
with
\begin{equation}
	I_s(t) = -\frac{1}{2} \tr{[\rho^{s}(t),H][\rho^{1-s}(t),H]},
\end{equation}
where $s \in (0,1)$. The quantity $I_s(t)$ was introduced as a measure of the information contained in measurements that do not commute with $H$~\cite{wigner1963information} (for $s=1/2$ one recovers the usual Wigner-Yanase skew information). A systematic framework to deal with all these constructions was developed in Ref.~\cite{gour2018quantum}. Leveraging the results obtained there, one can formally define a complete set of monotones. The drawback is that these form an infinite number of extremely involved conditions, and it is not yet clear how these can be simplified. 

An alternative approach is to obtain specific computable constraints. For example, exploiting properties~\ref{p1}-\ref{p2} and the framework of Ref.~\cite{styliaris2020symmetries} (see in particular Eq.~(22) therein), one can obtain the following one parameter family of monotones:
\begin{equation}
	\frac{d\mathcal{C}_\lambda(t)}{dt} \geq 0, \quad \mathcal{C}_\lambda(t) = \frac{-Z (E_j - E_i)^2}{\lambda e^{-\beta E_i} + e^{-\beta E_j}} |C_{ij}(t)|^2,
\end{equation}
where $C_{ij}(t)$ is the coherence matrix defined in Eq.~\eqref{eq:population_coherence2} and $\lambda \geq 0$. Note how each coherence element $|C_{ij}(t)|^2$ is weighted by a Gibbs factor and the energy difference squared, neatly combining energetic and coherent considerations. 

The general problem of Eq.~\eqref{eq:mtp} can be explicitly solved for a single qubit system using the minimal decoherence theory of Ref.~\cite{lostaglio2017markovian} (see Sec.~E3 therein). However, for higher dimensional systems currently there are no tools to obtain a solution to this problem (perhaps up to some approximation) that satisfies both our fundamental desiderata. We leave this extremely challenging question to future work.

%------------------------------------------------------------------

\bibliography{Bibliography}

\begin{thebibliography}{74}%
\makeatletter
\providecommand \@ifxundefined [1]{%
 \@ifx{#1\undefined}
}%
\providecommand \@ifnum [1]{%
 \ifnum #1\expandafter \@firstoftwo
 \else \expandafter \@secondoftwo
 \fi
}%
\providecommand \@ifx [1]{%
 \ifx #1\expandafter \@firstoftwo
 \else \expandafter \@secondoftwo
 \fi
}%
\providecommand \natexlab [1]{#1}%
\providecommand \enquote  [1]{``#1''}%
\providecommand \bibnamefont  [1]{#1}%
\providecommand \bibfnamefont [1]{#1}%
\providecommand \citenamefont [1]{#1}%
\providecommand \href@noop [0]{\@secondoftwo}%
\providecommand \href [0]{\begingroup \@sanitize@url \@href}%
\providecommand \@href[1]{\@@startlink{#1}\@@href}%
\providecommand \@@href[1]{\endgroup#1\@@endlink}%
\providecommand \@sanitize@url [0]{\catcode `\\12\catcode `\$12\catcode
  `\&12\catcode `\#12\catcode `\^12\catcode `\_12\catcode `\%12\relax}%
\providecommand \@@startlink[1]{}%
\providecommand \@@endlink[0]{}%
\providecommand \url  [0]{\begingroup\@sanitize@url \@url }%
\providecommand \@url [1]{\endgroup\@href {#1}{\urlprefix }}%
\providecommand \urlprefix  [0]{URL }%
\providecommand \Eprint [0]{\href }%
\providecommand \doibase [0]{https://doi.org/}%
\providecommand \selectlanguage [0]{\@gobble}%
\providecommand \bibinfo  [0]{\@secondoftwo}%
\providecommand \bibfield  [0]{\@secondoftwo}%
\providecommand \translation [1]{[#1]}%
\providecommand \BibitemOpen [0]{}%
\providecommand \bibitemStop [0]{}%
\providecommand \bibitemNoStop [0]{.\EOS\space}%
\providecommand \EOS [0]{\spacefactor3000\relax}%
\providecommand \BibitemShut  [1]{\csname bibitem#1\endcsname}%
\let\auto@bib@innerbib\@empty
%</preamble>
\bibitem [{\citenamefont {Lostaglio}(2019)}]{lostaglio2019introductory}%
  \BibitemOpen
  \bibfield  {author} {\bibinfo {author} {\bibfnamefont {M.}~\bibnamefont
  {Lostaglio}},\ }\bibfield  {title} {\bibinfo {title} {An introductory review
  of the resource theory approach to thermodynamics},\ }\href
  {https://doi.org/10.1088/1361-6633/ab46e5} {\bibfield  {journal} {\bibinfo
  {journal} {Rep. Prog. Phys.}\ }\textbf {\bibinfo {volume} {82}},\ \bibinfo
  {pages} {114001} (\bibinfo {year} {2019})}\BibitemShut {NoStop}%
\bibitem [{\citenamefont {{Janzing}}\ \emph {et~al.}(2000)\citenamefont
  {{Janzing}}, \citenamefont {{Wocjan}}, \citenamefont {{Zeier}}, \citenamefont
  {{Geiss}},\ and\ \citenamefont {{Beth}}}]{janzing2000thermodynamic}%
  \BibitemOpen
  \bibfield  {author} {\bibinfo {author} {\bibfnamefont {D.}~\bibnamefont
  {{Janzing}}}, \bibinfo {author} {\bibfnamefont {P.}~\bibnamefont {{Wocjan}}},
  \bibinfo {author} {\bibfnamefont {R.}~\bibnamefont {{Zeier}}}, \bibinfo
  {author} {\bibfnamefont {R.}~\bibnamefont {{Geiss}}},\ and\ \bibinfo {author}
  {\bibfnamefont {T.}~\bibnamefont {{Beth}}},\ }\bibfield  {title} {\bibinfo
  {title} {{Thermodynamic cost of reliability and low temperatures: tightening
  Landauer's principle and the second law}},\ }\href
  {https://doi.org/10.1023/A:1026422630734} {\bibfield  {journal} {\bibinfo
  {journal} {Int. J. Theor. Phys.}\ }\textbf {\bibinfo {volume} {39}},\
  \bibinfo {pages} {2717} (\bibinfo {year} {2000})}\BibitemShut {NoStop}%
\bibitem [{\citenamefont {{Horodecki}}\ and\ \citenamefont
  {{Oppenheim}}(2013)}]{horodecki2013fundamental}%
  \BibitemOpen
  \bibfield  {author} {\bibinfo {author} {\bibfnamefont {M.}~\bibnamefont
  {{Horodecki}}}\ and\ \bibinfo {author} {\bibfnamefont {J.}~\bibnamefont
  {{Oppenheim}}},\ }\bibfield  {title} {\bibinfo {title} {{Fundamental
  limitations for quantum and nanoscale thermodynamics}},\ }\href
  {https://www.nature.com/articles/ncomms3059} {\bibfield  {journal} {\bibinfo
  {journal} {Nat. Commun.}\ }\textbf {\bibinfo {volume} {4}},\ \bibinfo {eid}
  {2059} (\bibinfo {year} {2013})}\BibitemShut {NoStop}%
\bibitem [{\citenamefont {{Brand\~ao}}\ \emph {et~al.}(2015)\citenamefont
  {{Brand\~ao}}, \citenamefont {{Horodecki}}, \citenamefont {{Ng}},
  \citenamefont {{Oppenheim}},\ and\ \citenamefont
  {{Wehner}}}]{brandao2013second}%
  \BibitemOpen
  \bibfield  {author} {\bibinfo {author} {\bibfnamefont {F.~G.~S.~L.}\
  \bibnamefont {{Brand\~ao}}}, \bibinfo {author} {\bibfnamefont
  {M.}~\bibnamefont {{Horodecki}}}, \bibinfo {author} {\bibfnamefont
  {N.~H.~Y.}\ \bibnamefont {{Ng}}}, \bibinfo {author} {\bibfnamefont
  {J.}~\bibnamefont {{Oppenheim}}},\ and\ \bibinfo {author} {\bibfnamefont
  {S.}~\bibnamefont {{Wehner}}},\ }\bibfield  {title} {\bibinfo {title} {{The
  second laws of quantum thermodynamics}},\ }\href
  {https://doi.org/10.1073/pnas.1411728112} {\bibfield  {journal} {\bibinfo
  {journal} {Proc. Natl. Acad. Sci. U.S.A.}\ }\textbf {\bibinfo {volume}
  {112}},\ \bibinfo {pages} {3275} (\bibinfo {year} {2015})}\BibitemShut
  {NoStop}%
\bibitem [{\citenamefont {Ng}\ and\ \citenamefont
  {Woods}(2018)}]{ng2018resource}%
  \BibitemOpen
  \bibfield  {author} {\bibinfo {author} {\bibfnamefont {N.~H.~Y.}\
  \bibnamefont {Ng}}\ and\ \bibinfo {author} {\bibfnamefont {M.~P.}\
  \bibnamefont {Woods}},\ }\bibfield  {title} {\bibinfo {title} {Resource
  theory of quantum thermodynamics: Thermal operations and second laws},\ }in\
  \href@noop {} {\emph {\bibinfo {booktitle} {Thermodynamics in the Quantum
  Regime}}}\ (\bibinfo  {publisher} {Springer},\ \bibinfo {year} {2018})\ pp.\
  \bibinfo {pages} {625--650}\BibitemShut {NoStop}%
\bibitem [{\citenamefont {Vinjanampathy}\ and\ \citenamefont
  {Anders}(2016)}]{vinjanampathy2016quantum}%
  \BibitemOpen
  \bibfield  {author} {\bibinfo {author} {\bibfnamefont {S.}~\bibnamefont
  {Vinjanampathy}}\ and\ \bibinfo {author} {\bibfnamefont {J.}~\bibnamefont
  {Anders}},\ }\bibfield  {title} {\bibinfo {title} {Quantum thermodynamics},\
  }\href {https://doi.org/10.1080/00107514.2016.1201896} {\bibfield  {journal}
  {\bibinfo  {journal} {Contemp. Phys.}\ }\textbf {\bibinfo {volume} {57}},\
  \bibinfo {pages} {545} (\bibinfo {year} {2016})}\BibitemShut {NoStop}%
\bibitem [{\citenamefont {Kosloff}(2013)}]{kosloff2013quantum}%
  \BibitemOpen
  \bibfield  {author} {\bibinfo {author} {\bibfnamefont {R.}~\bibnamefont
  {Kosloff}},\ }\bibfield  {title} {\bibinfo {title} {Quantum thermodynamics: A
  dynamical viewpoint},\ }\href {https://doi.org/10.3390/e15062100} {\bibfield
  {journal} {\bibinfo  {journal} {Entropy}\ }\textbf {\bibinfo {volume} {15}},\
  \bibinfo {pages} {2100} (\bibinfo {year} {2013})}\BibitemShut {NoStop}%
\bibitem [{\citenamefont {Ruch}\ \emph {et~al.}(1978)\citenamefont {Ruch},
  \citenamefont {Schranner},\ and\ \citenamefont {Seligman}}]{ruch1978mixing}%
  \BibitemOpen
  \bibfield  {author} {\bibinfo {author} {\bibfnamefont {E.}~\bibnamefont
  {Ruch}}, \bibinfo {author} {\bibfnamefont {R.}~\bibnamefont {Schranner}},\
  and\ \bibinfo {author} {\bibfnamefont {T.~H.}\ \bibnamefont {Seligman}},\
  }\bibfield  {title} {\bibinfo {title} {The mixing distance},\ }\href
  {https://aip.scitation.org/doi/10.1063/1.436364} {\bibfield  {journal}
  {\bibinfo  {journal} {J. Chem. Phys.}\ }\textbf {\bibinfo {volume} {69}}
  (\bibinfo {year} {1978})}\BibitemShut {NoStop}%
\bibitem [{\citenamefont {Landi}\ and\ \citenamefont
  {Paternostro}(2021)}]{landi2021irreversible}%
  \BibitemOpen
  \bibfield  {author} {\bibinfo {author} {\bibfnamefont {G.~T.}\ \bibnamefont
  {Landi}}\ and\ \bibinfo {author} {\bibfnamefont {M.}~\bibnamefont
  {Paternostro}},\ }\bibfield  {title} {\bibinfo {title} {Irreversible entropy
  production, from quantum to classical},\ }\href
  {https://doi.org/10.1103/RevModPhys.93.035008} {\bibfield  {journal}
  {\bibinfo  {journal} {Rev. Mod. Phys.}\ }\textbf {\bibinfo {volume} {93}},\
  \bibinfo {pages} {035008} (\bibinfo {year} {2021})}\BibitemShut {NoStop}%
\bibitem [{\citenamefont {Korzekwa}(2021)}]{korzekwa2021continuous}%
  \BibitemOpen
  \bibfield  {author} {\bibinfo {author} {\bibfnamefont {K.}~\bibnamefont
  {Korzekwa}},\ }\href
  {https://github.com/KorzekwaKamil/continuous_thermomajorisation} {\emph
  {\bibinfo {title} {Continuous thermomajorisation}}} (\bibinfo {year}
  {2021}),\ \bibinfo {note} {github.com/KorzekwaKamil/continuous\textunderscore
  thermomajorisation}\BibitemShut {NoStop}%
\bibitem [{\citenamefont {Korzekwa}\ and\ \citenamefont
  {Lostaglio}(2022)}]{korzekwa2022optimizing}%
  \BibitemOpen
  \bibfield  {author} {\bibinfo {author} {\bibfnamefont {K.}~\bibnamefont
  {Korzekwa}}\ and\ \bibinfo {author} {\bibfnamefont {M.}~\bibnamefont
  {Lostaglio}},\ }\bibfield  {title} {\bibinfo {title} {Optimizing
  thermalizations},\ }\href {https://doi.org/10.1103/PhysRevLett.129.040602}
  {\bibfield  {journal} {\bibinfo  {journal} {Phys. Rev. Lett.}\ }\textbf
  {\bibinfo {volume} {129}},\ \bibinfo {pages} {040602} (\bibinfo {year}
  {2022})}\BibitemShut {NoStop}%
\bibitem [{\citenamefont {Spaventa}\ \emph {et~al.}(2022)\citenamefont
  {Spaventa}, \citenamefont {Huelga},\ and\ \citenamefont
  {Plenio}}]{spaventa2021non}%
  \BibitemOpen
  \bibfield  {author} {\bibinfo {author} {\bibfnamefont {G.}~\bibnamefont
  {Spaventa}}, \bibinfo {author} {\bibfnamefont {S.~F.}\ \bibnamefont
  {Huelga}},\ and\ \bibinfo {author} {\bibfnamefont {M.~B.}\ \bibnamefont
  {Plenio}},\ }\bibfield  {title} {\bibinfo {title} {Capacity of
  non-{M}arkovianity to boost the efficiency of molecular switches},\ }\href
  {https://doi.org/10.1103/PhysRevA.105.012420} {\bibfield  {journal} {\bibinfo
   {journal} {Phys. Rev. A}\ }\textbf {\bibinfo {volume} {105}},\ \bibinfo
  {pages} {012420} (\bibinfo {year} {2022})}\BibitemShut {NoStop}%
\bibitem [{\citenamefont {Nielsen}\ and\ \citenamefont
  {Chuang}(2010)}]{nielsen2010quantum}%
  \BibitemOpen
  \bibfield  {author} {\bibinfo {author} {\bibfnamefont {M.~A.}\ \bibnamefont
  {Nielsen}}\ and\ \bibinfo {author} {\bibfnamefont {I.~L.}\ \bibnamefont
  {Chuang}},\ }\href@noop {} {\emph {\bibinfo {title} {{Quantum computation and
  quantum information}}}}\ (\bibinfo  {publisher} {Cambridge university
  press},\ \bibinfo {year} {2010})\BibitemShut {NoStop}%
\bibitem [{\citenamefont {{Breuer}}\ and\ \citenamefont
  {{Petruccione}}(2002)}]{breuer2002open}%
  \BibitemOpen
  \bibfield  {author} {\bibinfo {author} {\bibfnamefont {H.-P.}\ \bibnamefont
  {{Breuer}}}\ and\ \bibinfo {author} {\bibfnamefont {F.}~\bibnamefont
  {{Petruccione}}},\ }\href@noop {} {\emph {\bibinfo {title} {The theory of
  open quantum systems}}}\ (\bibinfo  {publisher} {Oxford University Press},\
  \bibinfo {year} {2002})\BibitemShut {NoStop}%
\bibitem [{\citenamefont {Alicki}\ and\ \citenamefont
  {Kosloff}(2018)}]{alicki2018introduction}%
  \BibitemOpen
  \bibfield  {author} {\bibinfo {author} {\bibfnamefont {R.}~\bibnamefont
  {Alicki}}\ and\ \bibinfo {author} {\bibfnamefont {R.}~\bibnamefont
  {Kosloff}},\ }\bibinfo {title} {Introduction to quantum thermodynamics:
  History and prospects},\ in\ \href
  {https://doi.org/10.1007/978-3-319-99046-0_1} {\emph {\bibinfo {booktitle}
  {Thermodynamics in the Quantum Regime: Fundamental Aspects and New
  Directions}}},\ \bibinfo {editor} {edited by\ \bibinfo {editor}
  {\bibfnamefont {F.}~\bibnamefont {Binder}}, \bibinfo {editor} {\bibfnamefont
  {L.~A.}\ \bibnamefont {Correa}}, \bibinfo {editor} {\bibfnamefont
  {C.}~\bibnamefont {Gogolin}}, \bibinfo {editor} {\bibfnamefont
  {J.}~\bibnamefont {Anders}},\ and\ \bibinfo {editor} {\bibfnamefont
  {G.}~\bibnamefont {Adesso}}}\ (\bibinfo  {publisher} {Springer International
  Publishing},\ \bibinfo {address} {Cham},\ \bibinfo {year} {2018})\ pp.\
  \bibinfo {pages} {1--33}\BibitemShut {NoStop}%
\bibitem [{\citenamefont {Gorini}\ \emph {et~al.}(1976)\citenamefont {Gorini},
  \citenamefont {Kossakowski},\ and\ \citenamefont
  {Sudarshan}}]{gorini1976completely}%
  \BibitemOpen
  \bibfield  {author} {\bibinfo {author} {\bibfnamefont {V.}~\bibnamefont
  {Gorini}}, \bibinfo {author} {\bibfnamefont {A.}~\bibnamefont
  {Kossakowski}},\ and\ \bibinfo {author} {\bibfnamefont {E.~C.~G.}\
  \bibnamefont {Sudarshan}},\ }\bibfield  {title} {\bibinfo {title}
  {{Completely positive dynamical semigroups of N-level systems}},\ }\href
  {https://doi.org/10.1063/1.522979} {\bibfield  {journal} {\bibinfo  {journal}
  {J. Math. Phys.}\ }\textbf {\bibinfo {volume} {17}},\ \bibinfo {pages} {821}
  (\bibinfo {year} {1976})}\BibitemShut {NoStop}%
\bibitem [{\citenamefont {Lindblad}(1976)}]{lindblad1976generators}%
  \BibitemOpen
  \bibfield  {author} {\bibinfo {author} {\bibfnamefont {G.}~\bibnamefont
  {Lindblad}},\ }\bibfield  {title} {\bibinfo {title} {{On the generators of
  quantum dynamical semigroups}},\ }\href {https://doi.org/10.1007/BF01608499}
  {\bibfield  {journal} {\bibinfo  {journal} {Commun. Math. Phys.}\ }\textbf
  {\bibinfo {volume} {48}},\ \bibinfo {pages} {119} (\bibinfo {year}
  {1976})}\BibitemShut {NoStop}%
\bibitem [{\citenamefont {Spohn}(1978)}]{spohn1978entropy}%
  \BibitemOpen
  \bibfield  {author} {\bibinfo {author} {\bibfnamefont {H.}~\bibnamefont
  {Spohn}},\ }\bibfield  {title} {\bibinfo {title} {Entropy production for
  quantum dynamical semigroups},\ }\href {https://doi.org/10.1063/1.523789}
  {\bibfield  {journal} {\bibinfo  {journal} {J. Math. Phys.}\ }\textbf
  {\bibinfo {volume} {19}},\ \bibinfo {pages} {1227} (\bibinfo {year}
  {1978})}\BibitemShut {NoStop}%
\bibitem [{\citenamefont {Uzdin}\ \emph {et~al.}(2015)\citenamefont {Uzdin},
  \citenamefont {Levy},\ and\ \citenamefont {Kosloff}}]{uzdin2015equivalence}%
  \BibitemOpen
  \bibfield  {author} {\bibinfo {author} {\bibfnamefont {R.}~\bibnamefont
  {Uzdin}}, \bibinfo {author} {\bibfnamefont {A.}~\bibnamefont {Levy}},\ and\
  \bibinfo {author} {\bibfnamefont {R.}~\bibnamefont {Kosloff}},\ }\bibfield
  {title} {\bibinfo {title} {Equivalence of quantum heat machines, and
  quantum-thermodynamic signatures},\ }\href
  {https://doi.org/10.1103/PhysRevX.5.031044} {\bibfield  {journal} {\bibinfo
  {journal} {Phys. Rev. X}\ }\textbf {\bibinfo {volume} {5}},\ \bibinfo {pages}
  {031044} (\bibinfo {year} {2015})}\BibitemShut {NoStop}%
\bibitem [{\citenamefont {Campbell}\ \emph {et~al.}(2017)\citenamefont
  {Campbell}, \citenamefont {Terhal},\ and\ \citenamefont
  {Vuillot}}]{campbell2017roads}%
  \BibitemOpen
  \bibfield  {author} {\bibinfo {author} {\bibfnamefont {E.~T.}\ \bibnamefont
  {Campbell}}, \bibinfo {author} {\bibfnamefont {B.~M.}\ \bibnamefont
  {Terhal}},\ and\ \bibinfo {author} {\bibfnamefont {C.}~\bibnamefont
  {Vuillot}},\ }\bibfield  {title} {\bibinfo {title} {Roads towards
  fault-tolerant universal quantum computation},\ }\href
  {https://www.nature.com/articles/nature23460} {\bibfield  {journal} {\bibinfo
   {journal} {Nature}\ }\textbf {\bibinfo {volume} {549}},\ \bibinfo {pages}
  {172} (\bibinfo {year} {2017})}\BibitemShut {NoStop}%
\bibitem [{\citenamefont {Klatzow}\ \emph {et~al.}(2019)\citenamefont
  {Klatzow}, \citenamefont {Becker}, \citenamefont {Ledingham}, \citenamefont
  {Weinzetl}, \citenamefont {Kaczmarek}, \citenamefont {Saunders},
  \citenamefont {Nunn}, \citenamefont {Walmsley}, \citenamefont {Uzdin},\ and\
  \citenamefont {Poem}}]{klatzow2019experimental}%
  \BibitemOpen
  \bibfield  {author} {\bibinfo {author} {\bibfnamefont {J.}~\bibnamefont
  {Klatzow}}, \bibinfo {author} {\bibfnamefont {J.~N.}\ \bibnamefont {Becker}},
  \bibinfo {author} {\bibfnamefont {P.~M.}\ \bibnamefont {Ledingham}}, \bibinfo
  {author} {\bibfnamefont {C.}~\bibnamefont {Weinzetl}}, \bibinfo {author}
  {\bibfnamefont {K.~T.}\ \bibnamefont {Kaczmarek}}, \bibinfo {author}
  {\bibfnamefont {D.~J.}\ \bibnamefont {Saunders}}, \bibinfo {author}
  {\bibfnamefont {J.}~\bibnamefont {Nunn}}, \bibinfo {author} {\bibfnamefont
  {I.~A.}\ \bibnamefont {Walmsley}}, \bibinfo {author} {\bibfnamefont
  {R.}~\bibnamefont {Uzdin}},\ and\ \bibinfo {author} {\bibfnamefont
  {E.}~\bibnamefont {Poem}},\ }\bibfield  {title} {\bibinfo {title}
  {Experimental demonstration of quantum effects in the operation of
  microscopic heat engines},\ }\href
  {https://doi.org/10.1103/PhysRevLett.122.110601} {\bibfield  {journal}
  {\bibinfo  {journal} {Phys. Rev. Lett.}\ }\textbf {\bibinfo {volume} {122}},\
  \bibinfo {pages} {110601} (\bibinfo {year} {2019})}\BibitemShut {NoStop}%
\bibitem [{\citenamefont {Gour}\ \emph {et~al.}(2018)\citenamefont {Gour},
  \citenamefont {Jennings}, \citenamefont {Buscemi}, \citenamefont {Duan},\
  and\ \citenamefont {Marvian}}]{gour2018quantum}%
  \BibitemOpen
  \bibfield  {author} {\bibinfo {author} {\bibfnamefont {G.}~\bibnamefont
  {Gour}}, \bibinfo {author} {\bibfnamefont {D.}~\bibnamefont {Jennings}},
  \bibinfo {author} {\bibfnamefont {F.}~\bibnamefont {Buscemi}}, \bibinfo
  {author} {\bibfnamefont {R.}~\bibnamefont {Duan}},\ and\ \bibinfo {author}
  {\bibfnamefont {I.}~\bibnamefont {Marvian}},\ }\bibfield  {title} {\bibinfo
  {title} {Quantum majorization and a complete set of entropic conditions for
  quantum thermodynamics},\ }\href
  {https://www.nature.com/articles/s41467-018-06261-7} {\bibfield  {journal}
  {\bibinfo  {journal} {Nat. Commun.}\ }\textbf {\bibinfo {volume} {9}},\
  \bibinfo {pages} {1} (\bibinfo {year} {2018})}\BibitemShut {NoStop}%
\bibitem [{\citenamefont {Brand\~ao}\ \emph {et~al.}(2013)\citenamefont
  {Brand\~ao}, \citenamefont {Horodecki}, \citenamefont {Oppenheim},
  \citenamefont {Renes},\ and\ \citenamefont {Spekkens}}]{brandao2011resource}%
  \BibitemOpen
  \bibfield  {author} {\bibinfo {author} {\bibfnamefont {F.~G. S.~L.}\
  \bibnamefont {Brand\~ao}}, \bibinfo {author} {\bibfnamefont {M.}~\bibnamefont
  {Horodecki}}, \bibinfo {author} {\bibfnamefont {J.}~\bibnamefont
  {Oppenheim}}, \bibinfo {author} {\bibfnamefont {J.~M.}\ \bibnamefont
  {Renes}},\ and\ \bibinfo {author} {\bibfnamefont {R.~W.}\ \bibnamefont
  {Spekkens}},\ }\bibfield  {title} {\bibinfo {title} {Resource theory of
  quantum states out of thermal equilibrium},\ }\href
  {https://doi.org/10.1103/PhysRevLett.111.250404} {\bibfield  {journal}
  {\bibinfo  {journal} {Phys. Rev. Lett.}\ }\textbf {\bibinfo {volume} {111}},\
  \bibinfo {pages} {250404} (\bibinfo {year} {2013})}\BibitemShut {NoStop}%
\bibitem [{\citenamefont {Perry}\ \emph {et~al.}(2018)\citenamefont {Perry},
  \citenamefont {\ifmmode \acute{C}\else \'{C}\fi{}wikli\ifmmode~\acute{n}\else
  \'{n}\fi{}ski}, \citenamefont {Anders}, \citenamefont {Horodecki},\ and\
  \citenamefont {Oppenheim}}]{perry2018sufficient}%
  \BibitemOpen
  \bibfield  {author} {\bibinfo {author} {\bibfnamefont {C.}~\bibnamefont
  {Perry}}, \bibinfo {author} {\bibfnamefont {P.}~\bibnamefont {\ifmmode
  \acute{C}\else \'{C}\fi{}wikli\ifmmode~\acute{n}\else \'{n}\fi{}ski}},
  \bibinfo {author} {\bibfnamefont {J.}~\bibnamefont {Anders}}, \bibinfo
  {author} {\bibfnamefont {M.}~\bibnamefont {Horodecki}},\ and\ \bibinfo
  {author} {\bibfnamefont {J.}~\bibnamefont {Oppenheim}},\ }\bibfield  {title}
  {\bibinfo {title} {A sufficient set of experimentally implementable thermal
  operations for small systems},\ }\href
  {https://doi.org/10.1103/PhysRevX.8.041049} {\bibfield  {journal} {\bibinfo
  {journal} {Phys. Rev. X}\ }\textbf {\bibinfo {volume} {8}},\ \bibinfo {pages}
  {041049} (\bibinfo {year} {2018})}\BibitemShut {NoStop}%
\bibitem [{\citenamefont {Vom~Ende}\ \emph {et~al.}(2019)\citenamefont
  {Vom~Ende}, \citenamefont {Dirr}, \citenamefont {Keyl},\ and\ \citenamefont
  {Schulte-Herbr{\"u}ggen}}]{vom2019reachability}%
  \BibitemOpen
  \bibfield  {author} {\bibinfo {author} {\bibfnamefont {F.}~\bibnamefont
  {Vom~Ende}}, \bibinfo {author} {\bibfnamefont {G.}~\bibnamefont {Dirr}},
  \bibinfo {author} {\bibfnamefont {M.}~\bibnamefont {Keyl}},\ and\ \bibinfo
  {author} {\bibfnamefont {T.}~\bibnamefont {Schulte-Herbr{\"u}ggen}},\
  }\bibfield  {title} {\bibinfo {title} {Reachability in infinite-dimensional
  unital open quantum systems with switchable gks--lindblad generators},\
  }\href {https://doi.org/10.1142/S1230161219500148} {\bibfield  {journal}
  {\bibinfo  {journal} {Open Syst. Inf. Dyn.}\ }\textbf {\bibinfo {volume}
  {26}},\ \bibinfo {pages} {1950014} (\bibinfo {year} {2019})}\BibitemShut
  {NoStop}%
\bibitem [{\citenamefont {Schulte-Herbr{\"u}ggen}\ \emph
  {et~al.}(2020)\citenamefont {Schulte-Herbr{\"u}ggen}, \citenamefont {Ende},\
  and\ \citenamefont {Dirr}}]{schulte2020exploring}%
  \BibitemOpen
  \bibfield  {author} {\bibinfo {author} {\bibfnamefont {T.}~\bibnamefont
  {Schulte-Herbr{\"u}ggen}}, \bibinfo {author} {\bibfnamefont {F.~v.}\
  \bibnamefont {Ende}},\ and\ \bibinfo {author} {\bibfnamefont
  {G.}~\bibnamefont {Dirr}},\ }\bibfield  {title} {\bibinfo {title} {Exploring
  the limits of open quantum dynamics {I}: Motivation, new results from toy
  models to applications},\ }\href {https://arxiv.org/abs/2003.06018}
  {\bibfield  {journal} {\bibinfo  {journal} {arXiv:2003.06018}\ } (\bibinfo
  {year} {2020})}\BibitemShut {NoStop}%
\bibitem [{\citenamefont {{Renyi}}(1961)}]{renyi1961measures}%
  \BibitemOpen
  \bibfield  {author} {\bibinfo {author} {\bibfnamefont {A.}~\bibnamefont
  {{Renyi}}},\ }\bibfield  {title} {\bibinfo {title} {{On measures of entropy
  and information}},\ }in\ \href@noop {} {\emph {\bibinfo {booktitle} {{Fourth
  Berkeley Symposium on Mathematical Statistics and Probability}}}}\ (\bibinfo
  {year} {1961})\ pp.\ \bibinfo {pages} {547--561}\BibitemShut {NoStop}%
\bibitem [{\citenamefont {Elfving}(1937)}]{elfving1937theorie}%
  \BibitemOpen
  \bibfield  {author} {\bibinfo {author} {\bibfnamefont {G.}~\bibnamefont
  {Elfving}},\ }\bibfield  {title} {\bibinfo {title} {Zur theorie der
  {M}arkoffschen ketten},\ }\href@noop {} {\bibfield  {journal} {\bibinfo
  {journal} {Acta Soc. Sci. Fennicae, n. Ser. A2}\ }\textbf {\bibinfo {volume}
  {8}},\ \bibinfo {pages} {1} (\bibinfo {year} {1937})}\BibitemShut {NoStop}%
\bibitem [{\citenamefont {Davies}(2010)}]{davies2010embeddable}%
  \BibitemOpen
  \bibfield  {author} {\bibinfo {author} {\bibfnamefont {E.~B.}\ \bibnamefont
  {Davies}},\ }\bibfield  {title} {\bibinfo {title} {{Embeddable Markov
  matrices}},\ }\href
  {https://projecteuclid.org/download/pdf_1/euclid.ejp/1464819832} {\bibfield
  {journal} {\bibinfo  {journal} {Electron. J. Probab.}\ }\textbf {\bibinfo
  {volume} {15}},\ \bibinfo {pages} {1474} (\bibinfo {year}
  {2010})}\BibitemShut {NoStop}%
\bibitem [{\citenamefont {Kingman}(1962)}]{kingman1962imbedding}%
  \BibitemOpen
  \bibfield  {author} {\bibinfo {author} {\bibfnamefont {J.~F.~C.}\
  \bibnamefont {Kingman}},\ }\bibfield  {title} {\bibinfo {title} {The
  imbedding problem for finite {M}arkov chains},\ }\href
  {https://link.springer.com/content/pdf/10.1007/BF00531768.pdf} {\bibfield
  {journal} {\bibinfo  {journal} {Probab. Theory Relat. Fields}\ }\textbf
  {\bibinfo {volume} {1}},\ \bibinfo {pages} {14} (\bibinfo {year}
  {1962})}\BibitemShut {NoStop}%
\bibitem [{\citenamefont {Runnenburg}(1962)}]{runnenburg1962elfving}%
  \BibitemOpen
  \bibfield  {author} {\bibinfo {author} {\bibfnamefont {J.~T.}\ \bibnamefont
  {Runnenburg}},\ }\bibfield  {title} {\bibinfo {title} {{On Elfving's problem
  of imbedding a time-discrete markov chain in a time-continuous one for
  finitely many states I}},\ }\href
  {https://doi.org/10.1016/S1385-7258(62)50052-8} {\bibfield  {journal}
  {\bibinfo  {journal} {P. K. Ned. Akad. A Math}\ }\textbf {\bibinfo {volume}
  {65}},\ \bibinfo {pages} {536} (\bibinfo {year} {1962})}\BibitemShut
  {NoStop}%
\bibitem [{\citenamefont {Goodman}(1970)}]{goodman1970intrinsic}%
  \BibitemOpen
  \bibfield  {author} {\bibinfo {author} {\bibfnamefont {G.~S.}\ \bibnamefont
  {Goodman}},\ }\bibfield  {title} {\bibinfo {title} {An intrinsic time for
  non-stationary finite {M}arkov chains},\ }\href
  {https://link.springer.com/content/pdf/10.1007/BF00534594.pdf} {\bibfield
  {journal} {\bibinfo  {journal} {Probab. Theory Relat. Fields}\ }\textbf
  {\bibinfo {volume} {16}},\ \bibinfo {pages} {165} (\bibinfo {year}
  {1970})}\BibitemShut {NoStop}%
\bibitem [{\citenamefont {Carette}(1995)}]{carette1995characterizations}%
  \BibitemOpen
  \bibfield  {author} {\bibinfo {author} {\bibfnamefont {P.}~\bibnamefont
  {Carette}},\ }\bibfield  {title} {\bibinfo {title} {Characterizations of
  embeddable $3\times 3$ stochastic matrices with a negative eigenvalue},\
  }\href {http://nyjm.albany.edu/j/1995/1-8p.pdf} {\bibfield  {journal}
  {\bibinfo  {journal} {New York J. Math}\ }\textbf {\bibinfo {volume} {1}},\
  \bibinfo {pages} {129} (\bibinfo {year} {1995})}\BibitemShut {NoStop}%
\bibitem [{\citenamefont {Marshall}\ \emph {et~al.}(2010)\citenamefont
  {Marshall}, \citenamefont {Olkin},\ and\ \citenamefont
  {Arnold}}]{marshall2010inequalities}%
  \BibitemOpen
  \bibfield  {author} {\bibinfo {author} {\bibfnamefont {A.~W.}\ \bibnamefont
  {Marshall}}, \bibinfo {author} {\bibfnamefont {I.}~\bibnamefont {Olkin}},\
  and\ \bibinfo {author} {\bibfnamefont {B.~C.}\ \bibnamefont {Arnold}},\
  }\href@noop {} {\emph {\bibinfo {title} {{Inequalities: Theory of
  Majorization and Its Applications}}}}\ (\bibinfo  {publisher} {Springer},\
  \bibinfo {year} {2010})\BibitemShut {NoStop}%
\bibitem [{\citenamefont {Nielsen}(2002)}]{nielsen2002introduction}%
  \BibitemOpen
  \bibfield  {author} {\bibinfo {author} {\bibfnamefont {M.~A.}\ \bibnamefont
  {Nielsen}},\ }\bibfield  {title} {\bibinfo {title} {An introduction to
  majorization and its applications to quantum mechanics},\ }\href
  {https://michaelnielsen.org/blog/talks/2002/maj/book.ps} {\bibfield
  {journal} {\bibinfo  {journal} {Lecture Notes, Department of Physics,
  University of Queensland, Australia}\ } (\bibinfo {year} {2002})}\BibitemShut
  {NoStop}%
\bibitem [{\citenamefont {Alhambra}\ \emph {et~al.}(2016)\citenamefont
  {Alhambra}, \citenamefont {Oppenheim},\ and\ \citenamefont
  {Perry}}]{alhambra2016fluctuating}%
  \BibitemOpen
  \bibfield  {author} {\bibinfo {author} {\bibfnamefont {{\'A}.~M.}\
  \bibnamefont {Alhambra}}, \bibinfo {author} {\bibfnamefont {J.}~\bibnamefont
  {Oppenheim}},\ and\ \bibinfo {author} {\bibfnamefont {C.}~\bibnamefont
  {Perry}},\ }\bibfield  {title} {\bibinfo {title} {Fluctuating states: What is
  the probability of a thermodynamical transition?},\ }\href
  {https://doi.org/10.1103/PhysRevX.6.041016} {\bibfield  {journal} {\bibinfo
  {journal} {Phys. Rev. X}\ }\textbf {\bibinfo {volume} {6}},\ \bibinfo {pages}
  {041016} (\bibinfo {year} {2016})}\BibitemShut {NoStop}%
\bibitem [{\citenamefont {Zylka}(1985)}]{zylka1985note}%
  \BibitemOpen
  \bibfield  {author} {\bibinfo {author} {\bibfnamefont {C.}~\bibnamefont
  {Zylka}},\ }\bibfield  {title} {\bibinfo {title} {A note on the attainability
  of states by equalizing processes},\ }\href
  {https://doi.org/https://doi.org/10.1007/BF00529057} {\bibfield  {journal}
  {\bibinfo  {journal} {Theoret. Chim. Acta}\ }\textbf {\bibinfo {volume}
  {68}},\ \bibinfo {pages} {363} (\bibinfo {year} {1985})}\BibitemShut
  {NoStop}%
\bibitem [{\citenamefont {Zylka}(1990)}]{zylka1990accessibility}%
  \BibitemOpen
  \bibfield  {author} {\bibinfo {author} {\bibfnamefont {C.}~\bibnamefont
  {Zylka}},\ }\bibfield  {title} {\bibinfo {title} {On the accessibility of
  states for systems with dissipative dynamics},\ }\href
  {https://doi.org/https://doi.org/10.1002/andp.19905020225} {\bibfield
  {journal} {\bibinfo  {journal} {Ann. Phys. (Berl.)}\ }\textbf {\bibinfo
  {volume} {502}},\ \bibinfo {pages} {268} (\bibinfo {year}
  {1990})}\BibitemShut {NoStop}%
\bibitem [{\citenamefont {Alberti}\ \emph {et~al.}(2008)\citenamefont
  {Alberti}, \citenamefont {Crell}, \citenamefont {Uhlmann},\ and\
  \citenamefont {Zylka}}]{alberti2008order}%
  \BibitemOpen
  \bibfield  {author} {\bibinfo {author} {\bibfnamefont {P.~M.}\ \bibnamefont
  {Alberti}}, \bibinfo {author} {\bibfnamefont {B.}~\bibnamefont {Crell}},
  \bibinfo {author} {\bibfnamefont {A.}~\bibnamefont {Uhlmann}},\ and\ \bibinfo
  {author} {\bibfnamefont {C.}~\bibnamefont {Zylka}},\ }\bibfield  {title}
  {\bibinfo {title} {Order structure (majorization) and irreversible
  processes},\ }\href
  {https://pdfs.semanticscholar.org/839c/a4dd2ebe32d163836da07eece68c715b6a8b.pdf}
  {\bibfield  {journal} {\bibinfo  {journal} {Vernetzte
  Wissenschaften--Crosslinks in Natural and Social Sciences, PJ Plath, E.-Chr.
  Hass, eds}\ ,\ \bibinfo {pages} {281}} (\bibinfo {year} {2008})}\BibitemShut
  {NoStop}%
\bibitem [{\citenamefont {Hay}\ \emph {et~al.}(2015)\citenamefont {Hay},
  \citenamefont {Schiff},\ and\ \citenamefont {Fisch}}]{hay2015maximal}%
  \BibitemOpen
  \bibfield  {author} {\bibinfo {author} {\bibfnamefont {M.~J.}\ \bibnamefont
  {Hay}}, \bibinfo {author} {\bibfnamefont {J.}~\bibnamefont {Schiff}},\ and\
  \bibinfo {author} {\bibfnamefont {N.~J.}\ \bibnamefont {Fisch}},\ }\bibfield
  {title} {\bibinfo {title} {Maximal energy extraction under discrete diffusive
  exchange},\ }\href {https://doi.org/http://dx.doi.org/10.1063/1.4933018}
  {\bibfield  {journal} {\bibinfo  {journal} {Phys. Plasmas}\ }\textbf
  {\bibinfo {volume} {22}},\ \bibinfo {pages} {102108} (\bibinfo {year}
  {2015})}\BibitemShut {NoStop}%
\bibitem [{\citenamefont {Hay}\ \emph {et~al.}(2017)\citenamefont {Hay},
  \citenamefont {Schiff},\ and\ \citenamefont {Fisch}}]{hay2017extreme}%
  \BibitemOpen
  \bibfield  {author} {\bibinfo {author} {\bibfnamefont {M.}~\bibnamefont
  {Hay}}, \bibinfo {author} {\bibfnamefont {J.}~\bibnamefont {Schiff}},\ and\
  \bibinfo {author} {\bibfnamefont {N.~J.}\ \bibnamefont {Fisch}},\ }\bibfield
  {title} {\bibinfo {title} {On extreme points of the diffusion polytope},\
  }\href {https://doi.org/http://dx.doi.org/10.1016/j.physa.2017.01.038}
  {\bibfield  {journal} {\bibinfo  {journal} {Physica A}\ }\textbf {\bibinfo
  {volume} {473}},\ \bibinfo {pages} {225} (\bibinfo {year}
  {2017})}\BibitemShut {NoStop}%
\bibitem [{\citenamefont {Thon}\ and\ \citenamefont
  {Wallace}(2004)}]{thon2004dalton}%
  \BibitemOpen
  \bibfield  {author} {\bibinfo {author} {\bibfnamefont {D.}~\bibnamefont
  {Thon}}\ and\ \bibinfo {author} {\bibfnamefont {S.~W.}\ \bibnamefont
  {Wallace}},\ }\bibfield  {title} {\bibinfo {title} {Dalton transfers,
  inequality and altruism},\ }\href
  {https://doi.org/https://doi.org/10.1007/s00355-003-0226-x} {\bibfield
  {journal} {\bibinfo  {journal} {Soc. Choice Welf.}\ }\textbf {\bibinfo
  {volume} {22}},\ \bibinfo {pages} {447} (\bibinfo {year} {2004})}\BibitemShut
  {NoStop}%
\bibitem [{\citenamefont {Gorban}(2013)}]{gorban2013thermodynamic}%
  \BibitemOpen
  \bibfield  {author} {\bibinfo {author} {\bibfnamefont {A.~N.}\ \bibnamefont
  {Gorban}},\ }\bibfield  {title} {\bibinfo {title} {Thermodynamic tree: The
  space of admissible paths},\ }\href {https://doi.org/10.1137/120866919}
  {\bibfield  {journal} {\bibinfo  {journal} {SIAM J. Appl. Dyn. Syst.}\
  }\textbf {\bibinfo {volume} {12}},\ \bibinfo {pages} {246} (\bibinfo {year}
  {2013})}\BibitemShut {NoStop}%
\bibitem [{\citenamefont {Santos}\ \emph {et~al.}(2019)\citenamefont {Santos},
  \citenamefont {C{\'e}leri}, \citenamefont {Landi},\ and\ \citenamefont
  {Paternostro}}]{santos2019role}%
  \BibitemOpen
  \bibfield  {author} {\bibinfo {author} {\bibfnamefont {J.~P.}\ \bibnamefont
  {Santos}}, \bibinfo {author} {\bibfnamefont {L.~C.}\ \bibnamefont
  {C{\'e}leri}}, \bibinfo {author} {\bibfnamefont {G.~T.}\ \bibnamefont
  {Landi}},\ and\ \bibinfo {author} {\bibfnamefont {M.}~\bibnamefont
  {Paternostro}},\ }\bibfield  {title} {\bibinfo {title} {The role of quantum
  coherence in non-equilibrium entropy production},\ }\href
  {https://doi.org/10.1038/s41534-019-0138-y} {\bibfield  {journal} {\bibinfo
  {journal} {npj Quantum Inf.}\ }\textbf {\bibinfo {volume} {5}},\ \bibinfo
  {pages} {1} (\bibinfo {year} {2019})}\BibitemShut {NoStop}%
\bibitem [{\citenamefont {Tsallis}(1988)}]{tsallis1988possible}%
  \BibitemOpen
  \bibfield  {author} {\bibinfo {author} {\bibfnamefont {C.}~\bibnamefont
  {Tsallis}},\ }\bibfield  {title} {\bibinfo {title} {Possible generalization
  of boltzmann-gibbs statistics},\ }\href {https://doi.org/10.1007/BF01016429}
  {\bibfield  {journal} {\bibinfo  {journal} {J. Stat. Phys.}\ }\textbf
  {\bibinfo {volume} {52}},\ \bibinfo {pages} {479} (\bibinfo {year}
  {1988})}\BibitemShut {NoStop}%
\bibitem [{\citenamefont {Abe}(2000)}]{abe2000axioms}%
  \BibitemOpen
  \bibfield  {author} {\bibinfo {author} {\bibfnamefont {S.}~\bibnamefont
  {Abe}},\ }\bibfield  {title} {\bibinfo {title} {Axioms and uniqueness theorem
  for tsallis entropy},\ }\href {https://doi.org/10.1016/S0375-9601(00)00337-6}
  {\bibfield  {journal} {\bibinfo  {journal} {Phys. Lett. A}\ }\textbf
  {\bibinfo {volume} {271}},\ \bibinfo {pages} {74} (\bibinfo {year}
  {2000})}\BibitemShut {NoStop}%
\bibitem [{\citenamefont {Tsallis}(2009)}]{tsallis2009introduction}%
  \BibitemOpen
  \bibfield  {author} {\bibinfo {author} {\bibfnamefont {C.}~\bibnamefont
  {Tsallis}},\ }\href@noop {} {\emph {\bibinfo {title} {Introduction to
  nonextensive statistical mechanics: approaching a complex world}}}\ (\bibinfo
   {publisher} {Springer Science \& Business Media},\ \bibinfo {year}
  {2009})\BibitemShut {NoStop}%
\bibitem [{\citenamefont {Mariz}(1992)}]{mariz1992irreversible}%
  \BibitemOpen
  \bibfield  {author} {\bibinfo {author} {\bibfnamefont {A.~M.}\ \bibnamefont
  {Mariz}},\ }\bibfield  {title} {\bibinfo {title} {On the irreversible nature
  of the tsallis and renyi entropies},\ }\href
  {https://doi.org/10.1016/0375-9601(92)90339-N} {\bibfield  {journal}
  {\bibinfo  {journal} {Phys. Lett. A}\ }\textbf {\bibinfo {volume} {165}},\
  \bibinfo {pages} {409} (\bibinfo {year} {1992})}\BibitemShut {NoStop}%
\bibitem [{\citenamefont {Wilming}\ and\ \citenamefont
  {Gallego}(2017)}]{wilming2017third}%
  \BibitemOpen
  \bibfield  {author} {\bibinfo {author} {\bibfnamefont {H.}~\bibnamefont
  {Wilming}}\ and\ \bibinfo {author} {\bibfnamefont {R.}~\bibnamefont
  {Gallego}},\ }\bibfield  {title} {\bibinfo {title} {Third law of
  thermodynamics as a single inequality},\ }\href
  {https://doi.org/10.1103/PhysRevX.7.041033} {\bibfield  {journal} {\bibinfo
  {journal} {Phys. Rev. X}\ }\textbf {\bibinfo {volume} {7}},\ \bibinfo {pages}
  {041033} (\bibinfo {year} {2017})}\BibitemShut {NoStop}%
\bibitem [{\citenamefont {Lostaglio}\ \emph {et~al.}(2018)\citenamefont
  {Lostaglio}, \citenamefont {Alhambra},\ and\ \citenamefont
  {Perry}}]{lostaglio2018elementary}%
  \BibitemOpen
  \bibfield  {author} {\bibinfo {author} {\bibfnamefont {M.}~\bibnamefont
  {Lostaglio}}, \bibinfo {author} {\bibfnamefont {{\'A}.~M.}\ \bibnamefont
  {Alhambra}},\ and\ \bibinfo {author} {\bibfnamefont {C.}~\bibnamefont
  {Perry}},\ }\bibfield  {title} {\bibinfo {title} {Elementary thermal
  operations},\ }\href {https://doi.org/10.22331/q-2018-02-08-52} {\bibfield
  {journal} {\bibinfo  {journal} {Quantum}\ }\textbf {\bibinfo {volume} {2}},\
  \bibinfo {pages} {52} (\bibinfo {year} {2018})}\BibitemShut {NoStop}%
\bibitem [{\citenamefont {Mazurek}\ and\ \citenamefont
  {Horodecki}(2018)}]{mazurek2018decomposability}%
  \BibitemOpen
  \bibfield  {author} {\bibinfo {author} {\bibfnamefont {P.}~\bibnamefont
  {Mazurek}}\ and\ \bibinfo {author} {\bibfnamefont {M.}~\bibnamefont
  {Horodecki}},\ }\bibfield  {title} {\bibinfo {title} {Decomposability and
  convex structure of thermal processes},\ }\href
  {https://doi.org/10.1088/1367-2630/aac057} {\bibfield  {journal} {\bibinfo
  {journal} {New J. Phys.}\ }\textbf {\bibinfo {volume} {20}},\ \bibinfo
  {pages} {053040} (\bibinfo {year} {2018})}\BibitemShut {NoStop}%
\bibitem [{\citenamefont {Davies}(1974)}]{davies1974markovian}%
  \BibitemOpen
  \bibfield  {author} {\bibinfo {author} {\bibfnamefont {E.~B.}\ \bibnamefont
  {Davies}},\ }\bibfield  {title} {\bibinfo {title} {{Markovian master
  equations}},\ }\href {https://doi.org/10.1007/BF01608389} {\bibfield
  {journal} {\bibinfo  {journal} {Commun. Math. Phys.}\ }\textbf {\bibinfo
  {volume} {39}},\ \bibinfo {pages} {91} (\bibinfo {year} {1974})}\BibitemShut
  {NoStop}%
\bibitem [{\citenamefont {Roga}\ \emph {et~al.}(2010)\citenamefont {Roga},
  \citenamefont {Fannes},\ and\ \citenamefont
  {{\.Z}yczkowski}}]{roga2010davies}%
  \BibitemOpen
  \bibfield  {author} {\bibinfo {author} {\bibfnamefont {W.}~\bibnamefont
  {Roga}}, \bibinfo {author} {\bibfnamefont {M.}~\bibnamefont {Fannes}},\ and\
  \bibinfo {author} {\bibfnamefont {K.}~\bibnamefont {{\.Z}yczkowski}},\
  }\bibfield  {title} {\bibinfo {title} {{Davies maps for qubits and
  qutrits}},\ }\href {https://doi.org/10.1016/S0034-4877(11)00003-6} {\bibfield
   {journal} {\bibinfo  {journal} {Rep. Math. Phys.}\ }\textbf {\bibinfo
  {volume} {66}},\ \bibinfo {pages} {311} (\bibinfo {year} {2010})}\BibitemShut
  {NoStop}%
\bibitem [{\citenamefont {Scarani}\ \emph {et~al.}(2002)\citenamefont
  {Scarani}, \citenamefont {Ziman}, \citenamefont {{\vv{S}}telmachovi{\vv{c}}},
  \citenamefont {Gisin},\ and\ \citenamefont
  {Bu{\vv{z}}ek}}]{scarani2002thermalizing}%
  \BibitemOpen
  \bibfield  {author} {\bibinfo {author} {\bibfnamefont {V.}~\bibnamefont
  {Scarani}}, \bibinfo {author} {\bibfnamefont {M.}~\bibnamefont {Ziman}},
  \bibinfo {author} {\bibfnamefont {P.}~\bibnamefont
  {{\vv{S}}telmachovi{\vv{c}}}}, \bibinfo {author} {\bibfnamefont
  {N.}~\bibnamefont {Gisin}},\ and\ \bibinfo {author} {\bibfnamefont
  {V.}~\bibnamefont {Bu{\vv{z}}ek}},\ }\bibfield  {title} {\bibinfo {title}
  {Thermalizing quantum machines: Dissipation and entanglement},\ }\href
  {https://doi.org/10.1103/PhysRevLett.88.097905} {\bibfield  {journal}
  {\bibinfo  {journal} {Phys. Rev. Lett.}\ }\textbf {\bibinfo {volume} {88}},\
  \bibinfo {pages} {097905} (\bibinfo {year} {2002})}\BibitemShut {NoStop}%
\bibitem [{\citenamefont {B{\"a}umer}\ \emph {et~al.}(2019)\citenamefont
  {B{\"a}umer}, \citenamefont {Perarnau-Llobet}, \citenamefont {Kammerlander},
  \citenamefont {Wilming},\ and\ \citenamefont {Renner}}]{baumer2019imperfect}%
  \BibitemOpen
  \bibfield  {author} {\bibinfo {author} {\bibfnamefont {E.}~\bibnamefont
  {B{\"a}umer}}, \bibinfo {author} {\bibfnamefont {M.}~\bibnamefont
  {Perarnau-Llobet}}, \bibinfo {author} {\bibfnamefont {P.}~\bibnamefont
  {Kammerlander}}, \bibinfo {author} {\bibfnamefont {H.}~\bibnamefont
  {Wilming}},\ and\ \bibinfo {author} {\bibfnamefont {R.}~\bibnamefont
  {Renner}},\ }\bibfield  {title} {\bibinfo {title} {Imperfect thermalizations
  allow for optimal thermodynamic processes},\ }\href
  {https://doi.org/10.22331/q-2019-06-24-153} {\bibfield  {journal} {\bibinfo
  {journal} {Quantum}\ }\textbf {\bibinfo {volume} {3}},\ \bibinfo {pages}
  {153} (\bibinfo {year} {2019})}\BibitemShut {NoStop}%
\bibitem [{\citenamefont {Miller}\ \emph {et~al.}(2019)\citenamefont {Miller},
  \citenamefont {Scandi}, \citenamefont {Anders},\ and\ \citenamefont
  {Perarnau-Llobet}}]{miller2019work}%
  \BibitemOpen
  \bibfield  {author} {\bibinfo {author} {\bibfnamefont {H.~J.}\ \bibnamefont
  {Miller}}, \bibinfo {author} {\bibfnamefont {M.}~\bibnamefont {Scandi}},
  \bibinfo {author} {\bibfnamefont {J.}~\bibnamefont {Anders}},\ and\ \bibinfo
  {author} {\bibfnamefont {M.}~\bibnamefont {Perarnau-Llobet}},\ }\bibfield
  {title} {\bibinfo {title} {Work fluctuations in slow processes: quantum
  signatures and optimal control},\ }\href
  {https://doi.org/10.1103/PhysRevLett.123.230603} {\bibfield  {journal}
  {\bibinfo  {journal} {Phys. Rev. Lett.}\ }\textbf {\bibinfo {volume} {123}},\
  \bibinfo {pages} {230603} (\bibinfo {year} {2019})}\BibitemShut {NoStop}%
\bibitem [{\citenamefont {Bhatia}(2013)}]{bhatia2013matrix}%
  \BibitemOpen
  \bibfield  {author} {\bibinfo {author} {\bibfnamefont {R.}~\bibnamefont
  {Bhatia}},\ }\href@noop {} {\emph {\bibinfo {title} {Matrix analysis}}},\
  Vol.\ \bibinfo {volume} {169}\ (\bibinfo  {publisher} {Springer Science \&
  Business Media},\ \bibinfo {year} {2013})\BibitemShut {NoStop}%
\bibitem [{\citenamefont {Abiuso}\ \emph {et~al.}(2020)\citenamefont {Abiuso},
  \citenamefont {JD~Miller}, \citenamefont {Perarnau-Llobet},\ and\
  \citenamefont {Scandi}}]{abiuso2020geometric}%
  \BibitemOpen
  \bibfield  {author} {\bibinfo {author} {\bibfnamefont {P.}~\bibnamefont
  {Abiuso}}, \bibinfo {author} {\bibfnamefont {H.}~\bibnamefont {JD~Miller}},
  \bibinfo {author} {\bibfnamefont {M.}~\bibnamefont {Perarnau-Llobet}},\ and\
  \bibinfo {author} {\bibfnamefont {M.}~\bibnamefont {Scandi}},\ }\bibfield
  {title} {\bibinfo {title} {Geometric optimisation of quantum thermodynamic
  processes},\ }\href {https://doi.org/10.3390/e22101076} {\bibfield  {journal}
  {\bibinfo  {journal} {Entropy}\ }\textbf {\bibinfo {volume} {22}},\ \bibinfo
  {pages} {1076} (\bibinfo {year} {2020})}\BibitemShut {NoStop}%
\bibitem [{\citenamefont {Lostaglio}\ \emph
  {et~al.}(2015{\natexlab{a}})\citenamefont {Lostaglio}, \citenamefont
  {M\"uller},\ and\ \citenamefont {Pastena}}]{lostaglio2015stochastic}%
  \BibitemOpen
  \bibfield  {author} {\bibinfo {author} {\bibfnamefont {M.}~\bibnamefont
  {Lostaglio}}, \bibinfo {author} {\bibfnamefont {M.~P.}\ \bibnamefont
  {M\"uller}},\ and\ \bibinfo {author} {\bibfnamefont {M.}~\bibnamefont
  {Pastena}},\ }\bibfield  {title} {\bibinfo {title} {Stochastic independence
  as a resource in small-scale thermodynamics},\ }\href
  {https://doi.org/10.1103/PhysRevLett.115.150402} {\bibfield  {journal}
  {\bibinfo  {journal} {Phys. Rev. Lett.}\ }\textbf {\bibinfo {volume} {115}},\
  \bibinfo {pages} {150402} (\bibinfo {year} {2015}{\natexlab{a}})}\BibitemShut
  {NoStop}%
\bibitem [{\citenamefont {M\"uller}(2018)}]{mueller2018correlating}%
  \BibitemOpen
  \bibfield  {author} {\bibinfo {author} {\bibfnamefont {M.~P.}\ \bibnamefont
  {M\"uller}},\ }\bibfield  {title} {\bibinfo {title} {Correlating thermal
  machines and the second law at the nanoscale},\ }\href
  {https://doi.org/10.1103/PhysRevX.8.041051} {\bibfield  {journal} {\bibinfo
  {journal} {Phys. Rev. X}\ }\textbf {\bibinfo {volume} {8}},\ \bibinfo {pages}
  {041051} (\bibinfo {year} {2018})}\BibitemShut {NoStop}%
\bibitem [{\citenamefont {Alhambra}\ \emph {et~al.}(2019)\citenamefont
  {Alhambra}, \citenamefont {Lostaglio},\ and\ \citenamefont
  {Perry}}]{Alhambra2019heatbathalgorithmic}%
  \BibitemOpen
  \bibfield  {author} {\bibinfo {author} {\bibfnamefont {{\'{A}}.~M.}\
  \bibnamefont {Alhambra}}, \bibinfo {author} {\bibfnamefont {M.}~\bibnamefont
  {Lostaglio}},\ and\ \bibinfo {author} {\bibfnamefont {C.}~\bibnamefont
  {Perry}},\ }\bibfield  {title} {\bibinfo {title} {Heat-{B}ath {A}lgorithmic
  {C}ooling with optimal thermalization strategies},\ }\href
  {https://doi.org/10.22331/q-2019-09-23-188} {\bibfield  {journal} {\bibinfo
  {journal} {{Quantum}}\ }\textbf {\bibinfo {volume} {3}},\ \bibinfo {pages}
  {188} (\bibinfo {year} {2019})}\BibitemShut {NoStop}%
\bibitem [{\citenamefont {Shiraishi}\ and\ \citenamefont
  {Sagawa}(2021)}]{shiraishi2021quantum}%
  \BibitemOpen
  \bibfield  {author} {\bibinfo {author} {\bibfnamefont {N.}~\bibnamefont
  {Shiraishi}}\ and\ \bibinfo {author} {\bibfnamefont {T.}~\bibnamefont
  {Sagawa}},\ }\bibfield  {title} {\bibinfo {title} {Quantum thermodynamics of
  correlated-catalytic state conversion at small scale},\ }\href
  {https://doi.org/10.1103/PhysRevLett.126.150502} {\bibfield  {journal}
  {\bibinfo  {journal} {Phys. Rev. Lett.}\ }\textbf {\bibinfo {volume} {126}},\
  \bibinfo {pages} {150502} (\bibinfo {year} {2021})}\BibitemShut {NoStop}%
\bibitem [{\citenamefont {Lipka-Bartosik}\ and\ \citenamefont
  {Skrzypczyk}(2021)}]{lipka2021all}%
  \BibitemOpen
  \bibfield  {author} {\bibinfo {author} {\bibfnamefont {P.}~\bibnamefont
  {Lipka-Bartosik}}\ and\ \bibinfo {author} {\bibfnamefont {P.}~\bibnamefont
  {Skrzypczyk}},\ }\bibfield  {title} {\bibinfo {title} {All states are
  universal catalysts in quantum thermodynamics},\ }\href
  {https://doi.org/10.1103/PhysRevX.11.011061} {\bibfield  {journal} {\bibinfo
  {journal} {Phys. Rev. X}\ }\textbf {\bibinfo {volume} {11}},\ \bibinfo
  {pages} {011061} (\bibinfo {year} {2021})}\BibitemShut {NoStop}%
\bibitem [{\citenamefont {Henao}\ and\ \citenamefont
  {Uzdin}(2021)}]{henao2021catalytic}%
  \BibitemOpen
  \bibfield  {author} {\bibinfo {author} {\bibfnamefont {I.}~\bibnamefont
  {Henao}}\ and\ \bibinfo {author} {\bibfnamefont {R.}~\bibnamefont {Uzdin}},\
  }\bibfield  {title} {\bibinfo {title} {Catalytic transformations with
  finite-size environments: applications to cooling and thermometry},\ }\href
  {https://doi.org/10.22331/q-2021-09-21-547} {\bibfield  {journal} {\bibinfo
  {journal} {Quantum}\ }\textbf {\bibinfo {volume} {5}},\ \bibinfo {pages}
  {547} (\bibinfo {year} {2021})}\BibitemShut {NoStop}%
\bibitem [{\citenamefont {Lostaglio}\ \emph
  {et~al.}(2015{\natexlab{b}})\citenamefont {Lostaglio}, \citenamefont
  {Jennings},\ and\ \citenamefont {Rudolph}}]{lostaglio2015description}%
  \BibitemOpen
  \bibfield  {author} {\bibinfo {author} {\bibfnamefont {M.}~\bibnamefont
  {Lostaglio}}, \bibinfo {author} {\bibfnamefont {D.}~\bibnamefont
  {Jennings}},\ and\ \bibinfo {author} {\bibfnamefont {T.}~\bibnamefont
  {Rudolph}},\ }\bibfield  {title} {\bibinfo {title} {Description of quantum
  coherence in thermodynamic processes requires constraints beyond free
  energy},\ }\href {http://dx.doi.org/10.1038/ncomms7383} {\bibfield  {journal}
  {\bibinfo  {journal} {Nat. Commun.}\ }\textbf {\bibinfo {volume} {6}},\
  \bibinfo {pages} {6383} (\bibinfo {year} {2015}{\natexlab{b}})}\BibitemShut
  {NoStop}%
\bibitem [{\citenamefont {{Marvian}}\ and\ \citenamefont
  {{Spekkens}}(2014)}]{marvian2014extending}%
  \BibitemOpen
  \bibfield  {author} {\bibinfo {author} {\bibfnamefont {I.}~\bibnamefont
  {{Marvian}}}\ and\ \bibinfo {author} {\bibfnamefont {R.~W.}\ \bibnamefont
  {{Spekkens}}},\ }\bibfield  {title} {\bibinfo {title} {{Extending {N}oether's
  theorem by quantifying the asymmetry of quantum states}},\ }\href
  {http://dx.doi.org/10.1038/ncomms4821} {\bibfield  {journal} {\bibinfo
  {journal} {Nat. Commun.}\ }\textbf {\bibinfo {volume} {5}},\ \bibinfo {pages}
  {3821} (\bibinfo {year} {2014})}\BibitemShut {NoStop}%
\bibitem [{\citenamefont {Mosonyi}\ and\ \citenamefont
  {Ogawa}(2015)}]{mosonyi2015quantum}%
  \BibitemOpen
  \bibfield  {author} {\bibinfo {author} {\bibfnamefont {M.}~\bibnamefont
  {Mosonyi}}\ and\ \bibinfo {author} {\bibfnamefont {T.}~\bibnamefont
  {Ogawa}},\ }\bibfield  {title} {\bibinfo {title} {Quantum hypothesis testing
  and the operational interpretation of the quantum r{\'e}nyi relative
  entropies},\ }\href {https://doi.org/10.1007/s00220-014-2248-x} {\bibfield
  {journal} {\bibinfo  {journal} {Commun. Math. Phys.}\ }\textbf {\bibinfo
  {volume} {334}},\ \bibinfo {pages} {1617} (\bibinfo {year}
  {2015})}\BibitemShut {NoStop}%
\bibitem [{\citenamefont {Marvian}(2012)}]{marvianthesis}%
  \BibitemOpen
  \bibfield  {author} {\bibinfo {author} {\bibfnamefont {I.}~\bibnamefont
  {Marvian}},\ }\emph {\bibinfo {title} {Symmetry, Asymmetry and Quantum
  Information}},\ \href@noop {} {Ph.D. thesis},\ \bibinfo  {school} {University
  of Waterloo} (\bibinfo {year} {2012})\BibitemShut {NoStop}%
\bibitem [{\citenamefont {B{\"a}umer}\ \emph {et~al.}(2018)\citenamefont
  {B{\"a}umer}, \citenamefont {Lostaglio}, \citenamefont {Perarnau-Llobet},\
  and\ \citenamefont {Sampaio}}]{baumer2018fluctuating}%
  \BibitemOpen
  \bibfield  {author} {\bibinfo {author} {\bibfnamefont {E.}~\bibnamefont
  {B{\"a}umer}}, \bibinfo {author} {\bibfnamefont {M.}~\bibnamefont
  {Lostaglio}}, \bibinfo {author} {\bibfnamefont {M.}~\bibnamefont
  {Perarnau-Llobet}},\ and\ \bibinfo {author} {\bibfnamefont {R.}~\bibnamefont
  {Sampaio}},\ }\bibfield  {title} {\bibinfo {title} {Fluctuating work in
  coherent quantum systems: Proposals and limitations},\ }in\ \href
  {https://doi.org/10.1007/978-3-319-99046-0_11} {\emph {\bibinfo {booktitle}
  {Thermodynamics in the Quantum Regime}}}\ (\bibinfo  {publisher} {Springer},\
  \bibinfo {year} {2018})\ pp.\ \bibinfo {pages} {275--300}\BibitemShut
  {NoStop}%
\bibitem [{\citenamefont {Janzing}\ and\ \citenamefont
  {Beth}(2003)}]{janzing2003quasi}%
  \BibitemOpen
  \bibfield  {author} {\bibinfo {author} {\bibfnamefont {D.}~\bibnamefont
  {Janzing}}\ and\ \bibinfo {author} {\bibfnamefont {T.}~\bibnamefont {Beth}},\
  }\bibfield  {title} {\bibinfo {title} {Quasi-order of clocks and their
  synchronism and quantum bounds for copying timing information},\ }\href
  {https://doi.org/10.1109/TIT.2002.806162} {\bibfield  {journal} {\bibinfo
  {journal} {IEEE Trans. Inf. Theory}\ }\textbf {\bibinfo {volume} {49}},\
  \bibinfo {pages} {230} (\bibinfo {year} {2003})}\BibitemShut {NoStop}%
\bibitem [{\citenamefont {Kwon}\ \emph {et~al.}(2018)\citenamefont {Kwon},
  \citenamefont {Jeong}, \citenamefont {Jennings}, \citenamefont {Yadin},\ and\
  \citenamefont {Kim}}]{kwon2018clock}%
  \BibitemOpen
  \bibfield  {author} {\bibinfo {author} {\bibfnamefont {H.}~\bibnamefont
  {Kwon}}, \bibinfo {author} {\bibfnamefont {H.}~\bibnamefont {Jeong}},
  \bibinfo {author} {\bibfnamefont {D.}~\bibnamefont {Jennings}}, \bibinfo
  {author} {\bibfnamefont {B.}~\bibnamefont {Yadin}},\ and\ \bibinfo {author}
  {\bibfnamefont {M.~S.}\ \bibnamefont {Kim}},\ }\bibfield  {title} {\bibinfo
  {title} {Clock--work trade-off relation for coherence in quantum
  thermodynamics},\ }\href {https://doi.org/10.1103/PhysRevLett.120.150602}
  {\bibfield  {journal} {\bibinfo  {journal} {Phys. Rev. Lett.}\ }\textbf
  {\bibinfo {volume} {120}},\ \bibinfo {pages} {150602} (\bibinfo {year}
  {2018})}\BibitemShut {NoStop}%
\bibitem [{\citenamefont {Wigner}\ and\ \citenamefont
  {Yanase}(1963)}]{wigner1963information}%
  \BibitemOpen
  \bibfield  {author} {\bibinfo {author} {\bibfnamefont {E.}~\bibnamefont
  {Wigner}}\ and\ \bibinfo {author} {\bibfnamefont {M.~M.}\ \bibnamefont
  {Yanase}},\ }\bibfield  {title} {\bibinfo {title} {Information contents of
  distributions},\ }\href {https://www.jstor.org/stable/71797} {\bibfield
  {journal} {\bibinfo  {journal} {Proc. Natl. Acad. Sci. U.S.A.}\ ,\ \bibinfo
  {pages} {910}} (\bibinfo {year} {1963})}\BibitemShut {NoStop}%
\bibitem [{\citenamefont {Styliaris}\ and\ \citenamefont
  {Zanardi}(2020)}]{styliaris2020symmetries}%
  \BibitemOpen
  \bibfield  {author} {\bibinfo {author} {\bibfnamefont {G.}~\bibnamefont
  {Styliaris}}\ and\ \bibinfo {author} {\bibfnamefont {P.}~\bibnamefont
  {Zanardi}},\ }\bibfield  {title} {\bibinfo {title} {Symmetries and monotones
  in markovian quantum dynamics},\ }\href
  {https://doi.org/10.22331/q-2020-04-30-261} {\bibfield  {journal} {\bibinfo
  {journal} {Quantum}\ }\textbf {\bibinfo {volume} {4}},\ \bibinfo {pages}
  {261} (\bibinfo {year} {2020})}\BibitemShut {NoStop}%
\bibitem [{\citenamefont {Lostaglio}\ \emph {et~al.}(2017)\citenamefont
  {Lostaglio}, \citenamefont {Korzekwa},\ and\ \citenamefont
  {Milne}}]{lostaglio2017markovian}%
  \BibitemOpen
  \bibfield  {author} {\bibinfo {author} {\bibfnamefont {M.}~\bibnamefont
  {Lostaglio}}, \bibinfo {author} {\bibfnamefont {K.}~\bibnamefont
  {Korzekwa}},\ and\ \bibinfo {author} {\bibfnamefont {A.}~\bibnamefont
  {Milne}},\ }\bibfield  {title} {\bibinfo {title} {Markovian evolution of
  quantum coherence under symmetric dynamics},\ }\href
  {https://doi.org/10.1103/PhysRevA.96.032109} {\bibfield  {journal} {\bibinfo
  {journal} {Phys. Rev. A}\ }\textbf {\bibinfo {volume} {96}},\ \bibinfo
  {pages} {032109} (\bibinfo {year} {2017})}\BibitemShut {NoStop}%
\end{thebibliography}%

\end{document}